\documentclass[12pt]{article}
\usepackage[height=220mm,width=150mm]{geometry}
\usepackage{array}
\usepackage{color}
\usepackage{mathrsfs}
\usepackage{hyperref}
\hypersetup{hidelinks}
\usepackage{amssymb}
\usepackage{amsmath}
\usepackage{amsthm}
\usepackage{color}
\allowdisplaybreaks
\usepackage{graphicx}
\usepackage{subfigure}
\usepackage{tikz}
\usepackage{algorithm}
\usepackage{algorithmic}
\usetikzlibrary{shapes.arrows,chains,positioning}
\newtheorem{Theorem}{Theorem}[section]

\newtheorem{Definition}[Theorem]{Definition}

\newtheorem{Remark}[Theorem]{Remark}

\def\nocolor#1{}

\begin{document}

\title{SVD-Based Graph Fractional Fourier Transform on Directed Graphs and Its Application}


\author{
Lu Li, Muhan Wang, Haiye Huo\thanks{Corresponding author.}\\
Department of Mathematics, School of Mathematics and Computer\\
Sciences, Nanchang University,
Nanchang 330031, Jiangxi, China\\
Emails: 13568007660@163.com; a2296201003@163.com;\\
hyhuo@ncu.edu.cn
}

\date{}
\maketitle

\textit{Abstract}.\,\,
Real-world signals frequently reside on directed Cartesian product graphs, including digital images, sensor networks, and meteorological temperature records. Designing a transform method suitable for processing such multi-dimensional graph signals within the fractional Fourier transform domain remains a critical challenge in graph signal processing (GSP). This paper proposes two novel graph fractional Fourier transforms (GFRFTs) for multi-dimensional signals defined on such directed product graphs and comprehensively investigates their denoising capabilities. Our contributions are fourfold: (1) We propose two distinct two-dimensional GFRFTs based on singular value decompositions of some fractional Laplacian matrices; (2) We generalize these transforms to multi-dimensional graph fractional Fourier transforms (MGFRFTs), establishing a powerful fractional domain analysis framework for multi-dimensional GSP; (3) We investigate the signal reconstruction capability of our proposed GFRFTs, as well as their computational complexity; and (4) We validate the practical utility of our approach through denoising experiments on real-world meteorological temperature datasets.

\textit{Keywords.}
Graph signal processing; graph fractional Fourier transform; directed graph; singular value decomposition; Cartesian product


\section{Introduction}\label{sec1}
	
    Graph signal processing (GSP) \cite{OFK21} can effectively process signals with irregular structures defined on graphs and has been widely used in sensor networks \cite{Jablonski17}, machine learning \cite{DTT20}, brain network function analysis \cite{HBM18}, and smart grid \cite{RS21}, etc. One of the fundamental tools in GSP is called graph Fourier transform (GFT) \cite{LA21}, which provides a frequency interpretation of graph signals.
	
    The GFT of undirected graphs has been extensively studied \cite{JXX24,SM14}. Let $\mathcal{G} = (\mathbf{V},\mathbf{E},\mathbf{A})$ be a weighted (un)directed graph. A traditional definition of the GFT on the undirected graph $\mathcal{G}$ is based on the eigendecomposition of the graph Laplacian $\mathbf{L}$ as follows:
    \begin{equation}\label{f1}
     \mathcal{F}\mathbf{x}:= \mathbf{U}^T\mathbf{x},
    \end{equation}
    where $\mathbf{x}$ denotes a signal on the undirected graph $\mathcal{G}$, and the orthogonal matrix $\mathbf{U} = [\mathbf{u}_0,\ldots,\mathbf{u}_{N-1}]$ is composed of the eigenvectors of $\mathbf{L}$. Here, the eigendecomposition of $\mathbf{L}$ is expressed as
    \begin{equation}
    \mathbf{L} = \mathbf{U}\boldsymbol{\Lambda}\mathbf{U}^T = \sum_{k=0}^{N-1} \lambda_k \mathbf{u}_k \mathbf{u}_k^T,
    \end{equation}
   where $\mathbf{\Lambda}$ is a diagonal matrix composed of eigenvalues $\lambda_k$ arranged in ascending order. Since the Laplacian matrix of a directed graph generally does not have eigendecomposition, the GFT defined in (\ref{f1}) is not applicable to directed graphs. Recently, several methods have been proposed for defining the GFT on directed graph \cite{SRB23,GAR19,CMZ18,ASG20,SS17}.

   For directed graphs, one approach is to perform a Jordan decomposition on the Laplacian operator $\mathbf{L}$ \cite{SM13}:
\begin{equation}
\mathbf{L} = \mathbf{U}\mathbf{J}\mathbf{U}^{-1},
\end{equation}
where $\mathbf{U}$ represents a nonsingular matrix and $\mathbf{J}$ is the Jordan canonical form of $\mathbf{L}$.
Then, the GFT of a signal $\mathbf{x}$ on the directed graph $\mathcal{G}$ is defined as
\begin{equation}
\mathcal{F}\mathbf{x} = \mathbf{U}^{-1}\mathbf{x}.
\end{equation}
However, this decomposition typically suffers from numerical instability, high computational cost, and does not satisfy Parseval's identity.

 Another approach involves constructing a Hermitian Laplacian operator $\mathbf{L}_q$ (with $0\leq q<1$) \cite{FSA20} that is positive semidefinite, thereby enabling its eigendecomposition:
\begin{equation}\label{lqd}
\mathbf{L}_q = \mathbf{U}_q \boldsymbol{\Lambda}_q \mathbf{U}_q^{*},
\end{equation}
where $\mathbf{U}_{q}=[\mathbf{u}_{q,0},\mathbf{u}_{q,1},\cdots,\mathbf{u}_{q,N-1}]$ is orthonormal, $\mathbf{\Lambda}_q={\rm{diag}}([\lambda_{q,0},\lambda_{q,1},\cdots,\\\lambda_{q,N-1}])$, and $^*$ represents conjugate transpose. Here, the eigenvalues of the Hermitian Laplacian matrix $\mathbf{L}_q$ are sorted in the ascending order, which satisfies $0\le \lambda_{q,0}\le \lambda_{q,1}\le\cdots\le \lambda_{q,N-1}$.
Based on $\mathbf{L}_q$ , the GFT of a signal $\mathbf{x}$ on a directed graph $\mathcal{G}$ is given by
\begin{equation}
\mathcal{F}_q \mathbf{x} = \mathbf{U}_q^{*}\mathbf{x}.
\end{equation}

Chen et al \cite{CCS23} proposed a new GFT for directed graphs $\mathcal{G}$ based on the singular value decomposition (SVD) of the Laplacian operator $\mathbf{L}$, expressed as
\begin{equation}
\mathbf{L}=\mathbf{U}\mathbf{\Sigma}\mathbf{V}^T=\sum_{k=0}^{N-1}\sigma_k \mathbf{u}_k \mathbf{v}_k^T.
\end{equation}
This approach offers numerical stability and computational efficiency.

To more effectively extract local features of graph signals, fractional order operators have been integrated into GSP \cite{YL23,GGW23}. Over the past two decades, fractional order has attracted considerable scholarly attention due to its significant applications in signal processing and optics. The graph fractional Fourier domain emerges from the interplay between the graph spectrum domain and fractional Fourier transform domain. However, research on graph fractional Fourier transform (GFRFT) remains relatively limited \cite{YL22,WY124,AKK24}.
For undirected graphs, Wang et al. \cite{WLC17} initially proposed a definition of the GFRFT based on the adjacency matrix, while Wu et al. \cite{WWY20} introduced the spectral graph fractional Fourier transform (SGFRFT) based on the fractional Laplacian matrix. For directed graphs, Yan et al. \cite{YL23} defined a directed graph fractional Fourier transform (DGFRFT) based on the Hermitian fractional Laplacian matrix. Although these transforms demonstrate excellent performance in revealing local features of graph signals, they suffer from limitations in displaying the spectrum of multi-dimensional graph signals, as they neglect the product structure within the graph.
	
Many engineering applications involve spatial-temporal data \cite{YYZ18,ILS16,HYK23,YKA22,LYK23,YPK25} that evolve over networks, such as measurements from sensor arrays collected across time. Such data are naturally modeled as signals on Cartesian product graphs, where the topology is represented by $\mathcal{T}\boxtimes \mathcal{S}$, combining a temporal graph $\mathcal{T}$ (e.g., a directed line graph for time series) with a spatial graph $\mathcal{S}$ (representing sensor connectivity). Although our theoretical framework is developed for general directed graphs $\mathcal{G}_1$ and $\mathcal{G}_2$ to ensure broad applicability, this spatial-temporal paradigm is fundamental to our problem formulation and algorithm design. For multi-dimensional signals defined on undirected Cartesian product graphs, Yan et al. \cite{YL22} proposed two multi-dimensional fractional Fourier transforms (MGFRFTs) based on fractional Laplacian matrix and adjacency matrix, respectively. To the best of our knowledge, there is currently no relevant research on GFRFT defined on directed Cartesian product graphs. Let $\mathcal{G}_{1}$ and $\mathcal{G}_{2}$ be two directed graphs of orders $N_1$ and $N_2$. The main contributions of this paper are as follows:
\begin{itemize}
    \item[(1)] Inspired by prior work \cite{CCL23} in GFT domain, we extend existing results to the fractional domain. Specifically, leveraging SVD of the fractional Laplacian matrix on the Cartesian product graph $\mathcal{G}_{1}\boxtimes\mathcal{G}_{2}$, and directed graphs $\mathcal{G}_{1}$ and $\mathcal{G}_{2}$, we introduce two novel GFRFTs on $\mathcal{G}_{1}\boxtimes\mathcal{G}_{2}$, denoted as $\mathcal{F}_{\boxtimes}^{\alpha}$ and $\mathcal{F}_{\otimes}^{\alpha}$ (see Definitions \ref{Def-mFRFT}, and \ref{Def:GFRFT2}). The GFRFT $\mathcal{F}_{\otimes}^{\alpha}$ and its inverse $\mathcal{F}_{\otimes}^{-\alpha}$ can be efficiently implemented via successive computations along the directions of $\mathcal{G}_{1}$, and $\mathcal{G}_{2}$ (see Algorithms \ref{algo1} and \ref{algo2}). The computational complexities of $\mathcal{F}_{\boxtimes}^{\alpha}$ and $\mathcal{F}_{\otimes}^{\alpha}$ are $\mathcal{O}(N_1^3N_2^3)$ and $\mathcal{O}(N_1^3+N_2^3)$, respectively.
    \item[(2)] For graph signals which are bandlimited in GFRFT domain, we prove that $\mathcal{F}_{\boxtimes}^{\alpha}$ and $\mathcal{F}_{\otimes}^{\alpha}$ both can provide accurate representations of signals exhibiting strong spatiotemporal correlations on the directed Cartesian product graph $\mathcal{G}_{1}\boxtimes\mathcal{G}_{2}$ (see Theorems \ref{Thm-GFRFT1} and \ref{Thm:An2}).
    \item[(3)] We extend the proposed two-dimensional GFRFTs $\mathcal{F}_{\boxtimes}^{\alpha}$ and $\mathcal{F}_{\otimes}^{\alpha}$ to the multi-dimensional case, as detailed in Section \ref{sec5}.
    \item[(4)] Numerical simulations on a meteorological dataset from the Brest region of France demonstrate that the proposed GFRFTs $\mathcal{F}_{\boxtimes}^{\alpha}$ and $\mathcal{F}_{\otimes}^{\alpha}$ exhibit excellent denoising performance, outperforming GFTs $\mathcal{F}_{\boxtimes}$ and $\mathcal{F}_{\otimes}$ proposed in \cite{CCL23} and GFRFT $F_{q}^{\alpha}$ based on the Hermitian fractional Laplacian matrix mentioned in \cite{YL23}.
  \end{itemize}

	The rest of this paper is organized as follows. In Section \ref{sec2}, we review preliminary information on two GFTs defined on directed graphs and GFRFT. In Sections \ref{sec3} and \ref{sec4}, we introduce two new types of GFRFTs $\mathcal{F}_{\boxtimes}^{\alpha}$ and $\mathcal{F}_{\otimes}^{\alpha}$ on Cartesian product graphs $\mathcal{G}=\mathcal{G}_{1}\boxtimes\mathcal{G}_{2}$ with directed graphs
	$\mathcal{G}_{1},\mathcal{G}_{2}$, respectively. Moreover, we prove that our proposed GFRFTs can express a graph signal defined on Cartesian product graph with strong spatiotemporal correlation efficiently.
	In Section \ref{sec5}, we extend the results obtained in Sections \ref{sec3} and \ref{sec4} to a Cartesian graph with $m$ directed graphs. In Section \ref{sec6}, we verify the denoising performance of our proposed GFRFTs
	$\mathcal{F}_{\boxtimes}^{\alpha}$ and $\mathcal{F}_{\otimes}^{\alpha}$ by simulation. In Section \ref{sec7}, we conclude the paper.
	
	\section{Preliminaries}\label{sec2}
	
	In this section, we briefly review some basic concepts of graph signals on directed graphs.
	
	\subsection{Cartesian product graph}
	
	Consider a weighted directed graph $\mathcal{G}=(\mathbf{V},\mathbf{E},\mathbf{A})$, where $\mathbf{V}=\{v_0,v_1,\cdots,v_{N-1}\}$ is the set of vertices with $N$ nodes in the graph, $\mathbf{E}$ is a set of edges with $\mathbf{E}=\{(i,j)|i,j\in \mathbf{V},j\thicksim i\}\subseteq \mathbf{V}\times \mathbf{V}$, and $\mathbf{A}$
	is the weighted adjacency matrix of the graph with entry $\mathbf{A}_{mn}=a_{mn}$ denotes the weight of the edge between two vertices $v_m$ and $v_n$.
	
	Given two directed graphs $\mathcal{G}_{1}=(\mathbf{V}_{1},\mathbf{E}_{1},\mathbf{A}_1)$ and $\mathcal{G}_{2}=(\mathbf{V}_{2}, \mathbf{E}_{2},\mathbf{A}_2)$, then
	$\mathcal{G}:=\mathcal{G}_{1}\boxtimes\mathcal{G}_{2}$ \cite{CCL23} represents the Cartesian product graph with vertex set $\mathbf{V}_{1}\times \mathbf{V}_{2}$, where the number of nodes in $\mathbf{V}_1$ and
	$\mathbf{V}_2$ are $N_1$ and $N_2$, respectively. The edge set of $\mathcal{G}_{1}\boxtimes\mathcal{G}_{2}$ satisfies
	\[
	[\{v_{1},\widetilde{v}_{1}\}\in \mathbf{E}_{1}, v_{2}=\widetilde{v}_{2}] \quad \mbox{or}\quad [v_{1}=\widetilde{v}_{1},\{v_{2},\widetilde{v}_{2}\}\in\mathbf{E}_{2}].
	\]
	Figure \ref{Cartesian} illustrates an example of a directed Cartesian product graph in GSP, where sensor network measurements are modeled as the Cartesian product of sensor network and time series. We assume that sensor work is a cycle graph, and time series is a path graph.

     \begin{figure}[htbp]
				\includegraphics[width=1\linewidth]{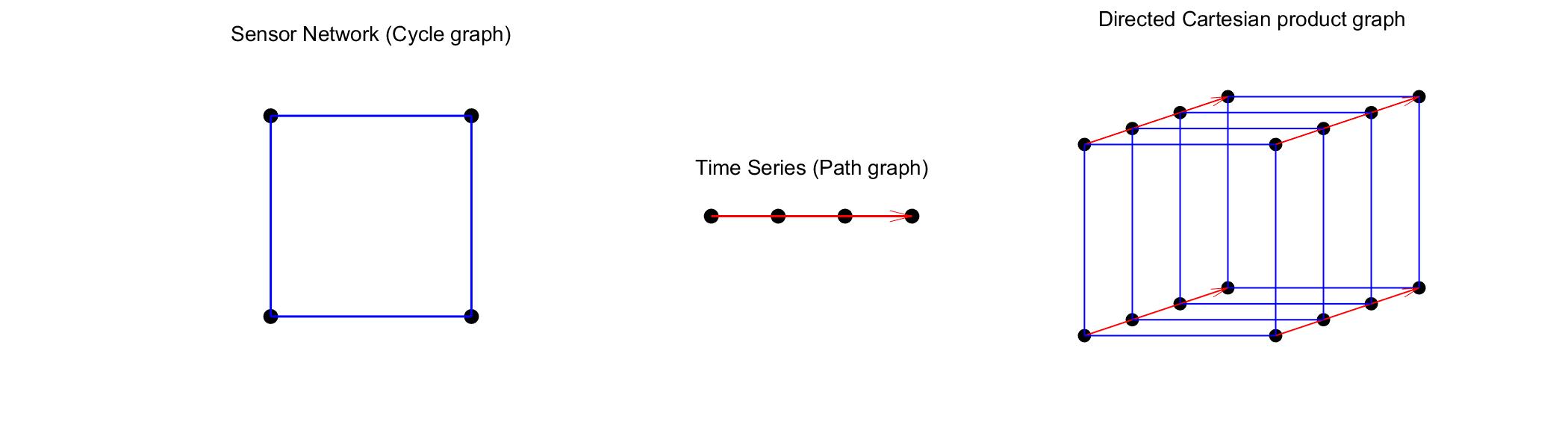}
		\centering
       \vspace {-1.0em}
		\caption{A directed Cartesian product graph.}\label{Cartesian}
	\end{figure}
	
For $l=1,2$, we define the degree matrix and the Laplacian matrix of graph $\mathcal{G}_{l}$ by $\mathbf{D}_{l}$ and $\mathbf{L}_{l}=\mathbf{D}_{l}-\mathbf{A}_{l}$, respectively. Then, the adjacency matrix
	$\mathbf{A}_{\boxtimes}$ and the Laplacian matrix $\mathbf{L}_{\boxtimes}$ of the Cartesian product graph $\mathcal{G}_{1}\boxtimes\mathcal{G}_{2}$ can be expressed as
	\begin{equation}\label{chuwz:1}
		\mathbf{A}_{\boxtimes}:=\mathbf{A}_1\oplus\mathbf{A}_2=\mathbf{A}_{1}\otimes\mathbf{I}_{N_{2}}+\mathbf{I}_{N_{1}}\otimes\textbf{A}_{2},
	\end{equation}
	and
	\begin{equation}\label{chuwz:2}
		\mathbf{L}_{\boxtimes}:=\mathbf{L}_1\oplus\mathbf{L}_2=\mathbf{L}_{1}\otimes\mathbf{I}_{N_{2}}+\mathbf{I}_{N_{1}}\otimes\mathbf{L}_{2},
	\end{equation}
	respectively. Here, the operator $\oplus$ represents the Kronecker sum, the operator $\otimes$ means the Kronecker product, and $\mathbf{I}_{N_i}$ denotes
	the identity matrix of size $N_i,\;i=1,2$.
	
	In the rest of this paper, we use an $N_2\times N_1$ matrix $\mathbf{X}=[\mathbf{x}_i]_{i\in \mathbf{V}_1}$ or its vectorization $\mathbf{x}={\rm{vec}}(\mathbf{X})$ to represent a signal defined on a Cartesian product graph $\mathcal{G}_{1}\boxtimes\mathcal{G}_{2}$, where $\mathbf{x}_i$ is a graph signal on $\mathcal{G}_2$, for all $i\in \mathbf{V}_1$. Alternatively, it can also be represented by a matrix $\mathbf{Y}=[\mathbf{y}_j^T]_{j\in \mathbf{V}_2}$, where $\mathbf{y}_j$ is a graph signal on $\mathcal{G}_1$, for all $j\in \mathbf{V}_2$. For spatial-temporal data, $\mathbf{x}_i$ is the spatial signal at time $i\in \mathbf{V}_1$ and $\mathbf{y}_j$ is the temporal signal at vertex $j\in \mathbf{V}_2$.

	\subsection{Graph Fourier transform on a directed Cartesian product graph}
	
	Most methods of defining GFT is essentially by decomposing a general graph shift operator. Cheng et al \cite{CCL23} defined two GFTs on directed Cartesian product graph by performing SVDs on the graph Laplacian matrices.
	
	For two directed graphs $\mathcal{G}_{1}=(\mathbf{V}_{1},\mathbf{E}_{1},\mathbf{A}_1)$ and $\mathcal{G}_{2}=(\mathbf{V}_{2}, \mathbf{E}_{2},\mathbf{A}_2)$, consider the directed Cartesian product graph
	$\mathcal{G}_{1}\boxtimes\mathcal{G}_{2}$. Assume that the singular values of the Laplacian matrix $\mathbf{L}_{\boxtimes}$ are sorted in a nondecreasing order
	$0=\sigma_{0}\leq\sigma_{1}\leq\cdots\leq\sigma_{N-1}$ with $N=N_{1}N_{2}$. Then, the SVD of the Laplacian matrix $\mathbf{L}_{\boxtimes}$ is a factorization of the form
	\begin{equation}\label{chuwz:3}
		\mathbf{L}_{\boxtimes}=\mathbf{U}_{\boxtimes}\mathbf{\Sigma}\mathbf{V}_{\boxtimes}^{\mathbf{T}}=\sum_{k=0}^{N-1}\sigma_{k}\mathbf{u}_{k}\mathbf{v}_{k}^{\mathbf{T}},
	\end{equation}
	where $\mathbf{\Sigma}={\rm{diag}}([\sigma_{0},\sigma_{1},\cdots,\sigma_{N-1}])$,  $\mathbf{U}_{\boxtimes}=[\mathbf{u}_{0},\mathbf{u}_{1},\cdots,\mathbf{u}_{N-1}]$ and $\mathbf{V}_{\boxtimes}=[\mathbf{v}_{0},\mathbf{v}_{1},\cdots,\\\mathbf{v}_{N-1}]$ are
	both orthogonal. Then, we can obtain the definition of GFT $\mathcal{F}_{\boxtimes}$ on directed Cartesian product graph $\mathcal{G}_{1}\boxtimes\mathcal{G}_{2}$ as follows:
	
\begin{Definition}\cite[Definition II.1]{CCL23}
	Given two directed graphs $\mathcal{G}_{1}$ and $\mathcal{G}_{2}$. Let $\mathbf{x}\in \mathbb{R}^N$ be a graph signal defined on a Cartesian product graph $\mathcal{G}_{1}\boxtimes\mathcal{G}_{2}$. Then, the GFT
	$\mathcal{F}_{\boxtimes}: \mathbb{R}^N\longmapsto \mathbb{R}^{2N}$ on $\mathcal{G}_{1}\boxtimes\mathcal{G}_{2}$ is given by
	\begin{eqnarray}\label{chuwz:4}
		\mathcal{F}_{\boxtimes}\mathbf{x}:=\dfrac{1}{2}\left(\begin{matrix}(\mathbf{U}_{\boxtimes}+\mathbf{V}_{\boxtimes})^{T}\mathbf{x}\\
			(\mathbf{U}_{\boxtimes}-\mathbf{V}_{\boxtimes})^{T}\mathbf{x} \end{matrix}\right)
		=\frac{1}{2}\left(\begin{matrix}(\mathbf{u}_{0}+\mathbf{v}_{0})^{T}\mathbf{x}\\
			\vdots\\
			(\mathbf{u}_{N-1}+\mathbf{v}_{N-1})^{T}\mathbf{x}\\
			(\mathbf{u}_{0}-\mathbf{v}_{0})^{T}\mathbf{x}\\
			\vdots\\
			(\mathbf{u}_{N-1}-\mathbf{v}_{N-1})^{T}\mathbf{x}\end{matrix}\right),
	\end{eqnarray}
		where $\mathbf{U}_{\boxtimes}$ and $\mathbf{V}_{\boxtimes}$ are defined the same as in $(\ref{chuwz:3})$. Moreover, for all
		$\mathbf{z}_{l}=[z_{l,0},z_{l,1},\cdots,\\z_{l,N-1}]^{T},\;l=1,2,$ the inverse GFT
		$\mathcal{F}_{\boxtimes}^{-1}:\mathbb{R}^{2N}\longmapsto \mathbb{R}^{N}$ is denoted as
		\begin{align}\label{chuwz:5}
			\mathcal{F}_{\boxtimes}^{-1}\left(\begin{matrix}\mathbf{z}_{1}\\ \mathbf{z}_{2}\end{matrix}\right)
			&:=\frac{1}{2}[\mathbf{U}_{\boxtimes}(\mathbf{z}_{1}+\mathbf{z}_{2})+\mathbf{V}_{\boxtimes}(\mathbf{z}_{1}-\mathbf{z}_{2})]\nonumber\\
			&=\frac{1}{2}\sum_{k=0}^{N-1}(z_{1,k}+z_{2,k})\mathbf{u}_{k}+(z_{1,k}-z_{2,k})\mathbf{v}_{k}.
		\end{align}
	\end{Definition}
	
	By performing SVDs on the Laplacian matrices $\mathbf{L}_{1}$ and $\mathbf{L}_{2}$ on two directed graphs $\mathcal{G}_{1}$ and $\mathcal{G}_{2}$, respectively, Cheng et al \cite{CCL23} also proposed another definition of GFT on a directed Cartesian product graph $\mathcal{G}_{1}\boxtimes\mathcal{G}_{2}$. Suppose that the singular values of Laplacian matrix $\mathbf{L}_l$ are sorted in a nondecreasing order
	$0=\sigma_{l,0}\leq\sigma_{l,1}\leq\cdots\leq\sigma_{l,N_l-1},\;\mbox{for}\;l=1,2$. Based on SVDs, the Laplacian matrices $\mathbf{L}_l,\;l=1,2$ can be decomposed as
	\begin{equation}\label{chuwz:6}
		\mathbf{L}_{l}=\mathbf{U}_{l}\mathbf{\Sigma}_{l}\mathbf{V}_{l}^{T}=\sum_{i=0}^{N_{l}-1}\sigma_{l,i}\mathbf{u}_{l,i}\mathbf{v}_{l,i}^{T},
	\end{equation}
	where $\mathbf{U}_{l}=[\mathbf{u}_{l,0},\mathbf{u}_{l,1},\cdots,\mathbf{u}_{l,N_{l}-1}]$ and $\mathbf{V}_{l}=[\mathbf{v}_{l,0},\mathbf{v}_{l,1},\cdots,\mathbf{v}_{l,N_{l}-1}]$
	are both orthogonal, $\mathbf{\Sigma}_l={\rm{diag}}([\sigma_{l,0},\sigma_{l,1},\cdots,\sigma_{l,{N_l}-1}])$. Denote
	\begin{equation}\label{gft1}
		\mathbf{U}_{\otimes}=\mathbf{U}_{1}\otimes\mathbf{U}_{2},\;\;\mathbf{V}_{\otimes}=\mathbf{V}_{1}\otimes\mathbf{V}_{2}.
	\end{equation}
	
	Then, another definition of GFT $\mathcal{F}_{\otimes}$ on a directed Cartesian product graph $\mathcal{G}_{1}\boxtimes\mathcal{G}_{2}$ is defined in the following:
	
	\begin{Definition}\cite[Definition III.1]{CCL23}
		For two directed graphs $\mathcal{G}_{1}$ and $\mathcal{G}_{2}$, assume that $\mathbf{x}\in \mathbb{R}^N$ is a graph signal defined on the Cartesian product graph $\mathcal{G}_{1}\boxtimes\mathcal{G}_{2}$.
		Then, the GFT $\mathcal{F}_{\otimes}: \mathbb{R}^N\longmapsto \mathbb{R}^{2N}$ on $\mathcal{G}_{1}\boxtimes\mathcal{G}_{2}$ can be represented as
		\begin{equation}\label{chuwz:7}
			\mathcal{F}_{\otimes}\mathbf{x}:=\frac{1}{2}\left(\begin{matrix}(\mathbf{U}_{\otimes}+\mathbf{V}_{\otimes})^{T}\mathbf{x}\\
				(\mathbf{U}_{\otimes}-\mathbf{V}_{\otimes})^{T}\mathbf{x} \end{matrix}\right),
		\end{equation}
		where $\mathbf{U}_{\otimes}$ and $\mathbf{V}_{\otimes}$ are defined the same as in $(\ref{gft1})$. Furthermore, for all $\mathbf{z}_{l}=[z_{l,0},z_{l,1},\cdots,z_{l,N-1}]^{T},\;l=1,2,$
		the inverse GFT
		$\mathcal{F}_{\otimes}^{-1}:\mathbb{R}^{2N}\longmapsto \mathbb{R}^{N}$ can be expressed as
		\begin{equation}\label{chuwz:8}
			\mathcal{F}_{\otimes}^{-1}\left(\begin{matrix}\mathbf{z}_{1}\\
				\mathbf{z}_{2}\end{matrix}\right):=\frac{1}{2}[\mathbf{U}_{\otimes}(\mathbf{z}_{1}+\mathbf{z}_{2})+\mathbf{V}_{\otimes}(\mathbf{z}_{1}-\mathbf{z}_{2})].
		\end{equation}
	\end{Definition}
	
	\subsection{Graph fractional Fourier Transform on a directed graph}\label{2.3}
	In this subsection, we review the concept of spectral graph fractional Fourier transform for a directed graph $\mathcal{G}=(\mathbf{V},\mathbf{E},\mathbf{A})$ (DGFRFT).
	
	Yan et al \cite{YL23} have extended the concept of Hermitian Laplacian matrix (\ref{lqd}) to fractional order. The graph Hermitian fractional Laplacian matrix of a directed graph $\mathcal{G}$
	is defined by
	\[
	\mathbf{L}_{q}^{\alpha}=\mathbf{P}_{q}\mathbf{\Upsilon}_{q}\mathbf{P}_{q}^*,
	\]
	with $0<\alpha\le 1$,
	\begin{equation}\label{eq:p}
		\mathbf{P}_{q}=[\mathbf{p}_{q,0},\mathbf{p}_{q,1},\cdots,\mathbf{p}_{q,N-1}]=\mathbf{U}_{q}^{\alpha}
	\end{equation}
	is an orthogonal matrix, and
	\begin{equation}
		\boldsymbol{\Upsilon}_{q}={\rm{diag}}([\mathbf{\varphi}_{q,0},\mathbf{\varphi}_{q,1},\cdots,\mathbf{\varphi}_{q,N-1}])=\mathbf{\Lambda}_{q}^{\alpha},
	\end{equation}
	i.e.,
	\begin{equation}
		\mathbf{\varphi}_{q,i}=\lambda_{q,i}^{\alpha},\;\mbox{for all}\; i=0,1,\cdots,N-1.
	\end{equation}
 In the subsequent sections of this paper, calculating the $\alpha$ power of a matrix refers to the matrix power function. This definition aims to construct a GFRFT basis using a new function $\mathbf{p}_{q,l}$, instead of sharing the same function as the GFT basis $\mathbf{u}_{q,l}$. By employing different basis functions, we ensure that GFRFT and GFT exhibit distinct characteristics.

	Then, Yan et al \cite{YL23} proposed a definition of DGFRFT for a directed graph $\mathcal{G}$ as follows.
	\begin{Definition}\cite[Definition 2]{YL23}\label{Def:DGFRFT}
		For any signal $\mathbf{x}$ defined on a directed graph $\mathcal{G}$, the DGFRFT is denoted as
		\begin{equation}\label{Def:DGFRFT2}
			\mathcal{F}_{q}^\alpha \mathbf{x}=\mathbf{P}_{q}^*\mathbf{x},
		\end{equation}
		where $\mathbf{P}_{q}$ is defined the same as in $(\ref{eq:p})$. Moreover, the inverse DGFRFT is defined by
		\begin{equation}\label{Def:DGFRFT3}
			\mathbf{x}=\mathbf{P}_{q}(\mathcal{F}_{q}^{\alpha}\mathbf{x}).
		\end{equation}
	\end{Definition}
	
	\section{SVD-Based GFRFT on Directed Cartesian Product Graph}\label{sec3}
	
GFTs $\mathcal{F}_{\boxtimes}$ and $\mathcal{F}_{\otimes}$ proposed in \cite{CCL23} are integer-order transforms, providing only two perspectives: 0 (time domain) or 1 (frequency domain). It lacks intermediate states, cannot characterize signal evolution in the mixed graph-frequency domain, nor flexibly adjust time-frequency resolution. Therefore, in this section, we introduce a GFRFT $\mathcal{F}_{\boxtimes}^{\alpha}$ on a directed Cartesian product graph $\mathcal{G}$ and prove in Theorem \ref{Thm-GFRFT1} that graph signals with strong spatiotemporal correlation concentrate their energy primarily in the low-frequency components of the proposed GFRFT.
	
	In the following, we consider a directed Cartesian product graph $\mathcal{G}_{1}\boxtimes\mathcal{G}_{2}$, where $\mathcal{G}_{1}=(\mathbf{V}_{1},\mathbf{E}_{1},\mathbf{A}_1)$ and $\mathcal{G}_{2}=(\mathbf{V}_{2}, \mathbf{E}_{2},\mathbf{A}_2)$ are two directed graphs. The SVDs of their Laplacian matrices $\mathbf{L}_l,\;l=1,2$ can be represented as
	\[
	\mathbf{L}_{l}=\mathbf{U}_{l}\mathbf{\Sigma}_{l}\mathbf{V}_{l}^{T},\; l=1,2,
	\]
	where $\mathbf{U}_{l},\;\mathbf{V}_{l},\;\mathbf{\Sigma}_{l}$ are defined the same as those in (\ref{chuwz:6}).
	
	For $0<\alpha\le 1$, the graph fractional Laplacian matrices $\mathbf{L}_{l}^{\alpha},l=1,2$ can be defined as:
	\begin{equation}\label{eq:La}
		\mathbf{L}^{\alpha}_{l}=\mathbf{P}_{l}\mathbf{R}_{l}\mathbf{Q}^{T}_{l},
	\end{equation}
	where
	\begin{align*}
		\mathbf{P}_{l}=[\mathbf{p}_{l,0},\mathbf{p}_{l,1},\cdots,\mathbf{p}_{l,{N_l}-1}]=\mathbf{U}_{l}^{\alpha},\quad
		\mathbf{Q}_{l}=[\mathbf{q}_{l,0},\mathbf{q}_{l,1},\cdots,\mathbf{q}_{l,{N_l}-1}]=\mathbf{V}_{l}^{\alpha},
	\end{align*}
	and
	\begin{equation*}
		\mathbf{R}_{l}={\rm{diag}}([r_{l,0},r_{l,1},\cdots,r_{l,N_l-1}])=\mathbf{\Sigma}^{\alpha}_{l},
	\end{equation*}
	which satisfies
	\begin{equation*}
		r_{l,i}=\sigma_{l,i}^{\alpha},\;i=0,1,\cdots,N_l-1.
	\end{equation*}
	
	Then, we define the graph fractional Laplacian matrix $\mathbf{L}_{\boxtimes}^{\alpha}$ for a directed Cartesian product graph $\mathcal{G}_{1}\boxtimes\mathcal{G}_{2}$ as
	\begin{equation}\label{def:frac}
		\mathbf{L}_{\boxtimes}^{\alpha}:=\mathbf{L}_1^{\alpha}\oplus\mathbf{L}_2^{\alpha}=\mathbf{L}_1^{\alpha}\otimes \mathbf{I}_{N_2}+\mathbf{I}_{N_1}\otimes\mathbf{L}_2^{\alpha}.
	\end{equation}
	
	By taking SVD on $\mathbf{L}_{\boxtimes}^{\alpha}$, it can be rewritten as
	\begin{equation}\label{chuwz:9}
		\mathbf{L}_{\boxtimes}^{\alpha}=\mathbf{P}_{\boxtimes}\mathbf{R}\mathbf{Q}_{\boxtimes}^{T}
		=\sum_{k=0}^{N-1}r_{k}\mathbf{p}_{k}\mathbf{q}_{k}^{T},
	\end{equation}
	where $N=N_1N_2$, two matrices
	\[
	\mathbf{P}_{\boxtimes}=[\mathbf{p}_{0},\mathbf{p}_{1},\cdots,\mathbf{p}_{N-1}],\;
	\mathbf{Q}_{\boxtimes}=[\mathbf{q}_{0},\mathbf{q}_{1},\cdots,\mathbf{q}_{N-1}]
	\]
	are orthonormal,
	\begin{equation*}
		\mathbf{R}={\rm{diag}}([r_{0},r_{1},\cdots,r_{N-1}]),
	\end{equation*}
	which satisfies $0=r_{0}\le r_{1}\le\cdots\le r_{N-1}$.
	The time complexity for computing the SVD factorization of $\mathbf{L}_{\boxtimes}^{\alpha}$ is $\mathcal{O}(N^{3})$.
	
	In particular, for the undirected graph case, that is, $\mathcal{G}_{1}$ and $\mathcal{G}_{2}$ are undirected graphs. Then, the graph fractional Laplacian matrices
	$\mathbf{L}_{l}^{\alpha},\; l=1,2$ are positive semi-definite, and can be represented as
	\begin{equation}\label{chuwz:10}
		\mathbf{L}_{l}^{\alpha}=\sum_{i=0}^{N_{l}-1}\rho_{l,i}\mathbf{k}_{l,i}\mathbf{k}_{l,i}^T,\; l=1,2,
	\end{equation}
	where $\{\rho_{l,i}\}_{i=0}^{N_{l}-1}$ are the eigenvalues of $\mathbf{L}_{l}^{\alpha}$ in an ascending order, and $\{\mathbf{k}_{l,i}\}_{i=0}^{N_{l}-1}$ are the eigenvectors. From matrix theory, it is known that the eigenvalues of the graph fractional Laplacian matrix $\mathbf{L}_{\boxtimes}^{\alpha}$ on an undirected Cartesian product graph $\mathcal{G}_{1}\boxtimes\mathcal{G}_{2}$ are equal to the sum of the eigenvalues of $\mathbf{L}_{1}^{\alpha}$ and $\mathbf{L}_{2}^{\alpha}$, and $\mathbf{P}_{\boxtimes}=\mathbf{Q}_{\boxtimes}$ is the Kronecker product of eigenfunctions of fractional Laplacian matrices $\mathbf{L}_{1}^{\alpha}$ and $\mathbf{L}_{2}^{\alpha}$, that is to say,
	\begin{equation}\label{chuwz:11}
		\mathbf{L}_{\boxtimes}^{\alpha}=\sum_{i=0}^{N_{1}-1}\sum_{j=0}^{N_{2}-1}(\rho_{1,i}+\rho_{2,j})(\mathbf{k}_{1,i}\otimes \mathbf{k}_{2,j})
		(\mathbf{k}_{1,i}\otimes \mathbf{k}_{2,j})^{T}.
	\end{equation}
	
	The computational complexity of performing the eigenvalue decomposition of the fractional Laplacian $\mathbf{L}_{\boxtimes}^{\alpha}$ is $\mathcal{O}(N_{1}^{3}+N_{2}^{3})$.
	
	Next, we extend the definition of GFT (\ref{chuwz:4}) introduced in \cite{CCL23} to the fractional order $\alpha$, where $0<\alpha\le 1$. The parameter $\alpha$ makes our GFRFT more flexible without extra computational cost.
	
\begin{Definition}\label{Def-mFRFT}
	Assume that $\mathcal{G}=\mathcal{G}_{1}\boxtimes\mathcal{G}_{2}$ is a Cartesian product graph of two directed graphs $\mathcal{G}_{1}$ and $\mathcal{G}_{2}$, the fractional Laplacian matrix $\mathbf{L}_{\boxtimes}^{\alpha}$ on $\mathcal{G}$ is defined the same as $(\ref{def:frac})$, and has the SVD form as in $(\ref{chuwz:9})$, $\alpha$ is the fractional order, which satisfies $0<\alpha\le 1$. The GFRFT $\mathcal{F}_{\boxtimes}^\alpha$ of a signal $\mathbf{x}:\mathbf{V}_1\times \mathbf{V}_2 \rightarrow \mathbb{R}^{N}$ on $\mathcal{G}$ is given by
	\begin{eqnarray}\label{chuwz:12}
		\mathcal{F}_{\boxtimes}^\alpha\mathbf{x} :=\dfrac{1}{2}\left(\begin{matrix}(\mathbf{P}_{\boxtimes}+\mathbf{Q}_{\boxtimes})^{T}\mathbf{x}\\
			(\mathbf{P}_{\boxtimes}-\mathbf{Q}_{\boxtimes})^{T}\mathbf{x} \end{matrix}\right)
		=\frac{1}{2}\left(\begin{matrix}(\mathbf{p}_{0}+\mathbf{q}_{0})^{T}\mathbf{x}\\
			\vdots\\
			(\mathbf{p}_{N-1}+\mathbf{q}_{N-1})^{T}\mathbf{x}\\
			(\mathbf{p}_{0}-\mathbf{q}_{0})^{T}\mathbf{x}\\
			\vdots\\
			(\mathbf{p}_{N-1}-\mathbf{q}_{N-1})^{T}\mathbf{x}\end{matrix}\right).
	\end{eqnarray}
		The inverse GFRFT $\mathcal{F}_{\boxtimes}^{-\alpha}$ is defined as
		\begin{align}
			\mathcal{F}_{\boxtimes}^{-\alpha}\left(\begin{matrix}\mathbf{y}_{1}\\ \mathbf{y}_{2}\end{matrix}\right)
			&:=\frac{1}{2}\left[\mathbf{P}_{\boxtimes}(\mathbf{y}_{1}+\mathbf{y}_{2})+\mathbf{Q}_{\boxtimes}(\mathbf{y}_{1}-\mathbf{y}_{2})\right]\nonumber\\
			&=\frac{1}{2}\sum_{i=0}^{N-1}\left[(y_{1,i}+y_{2,i})\mathbf{p}_{i}+(y_{1,i}
			-y_{2,i})\mathbf{q}_{i}\right],\label{chuwz:13}
		\end{align}
		for all $\mathbf{y}_l=[y_{l,0},y_{l,1},\cdots,y_{l, N-1}]^{T}\in \mathbb{R}^N,\;l=1,2.$
	\end{Definition}
	
	For a signal $\mathbf{x}$ on a Cartesian product graph $\mathcal{G}$, it is easy to prove that
	\begin{equation}\label{chuwz:15}
		\mathcal{F}_{\boxtimes}^{-\alpha}[\mathcal{F}_{\boxtimes}^\alpha\mathbf{x}]=\mathbf{x},
	\end{equation}
	and
	\begin{equation}\label{chuwz:16}
		\|\mathcal{F}_{\boxtimes}^\alpha\mathbf{x}\|_{2}=\|\mathbf{x}\|_{2},\;\; {\mbox{for all}}\;\; \mathbf{x}\in \mathbb{R}^N.
	\end{equation}
	
	When $\mathcal{G}_{1}$ and $\mathcal{G}_{2}$ are two undirected graphs, Yan et al $\cite{YL22}$ proposed a Laplacian-based multi-dimensional GFRFT $\mathcal{F}_{\alpha}$
	of a signal $\mathbf{X}$ as
	\begin{equation}\label{chuwz:14}
		\mathcal{F}_{\alpha}\mathbf{X}=\mathbf{K}_{2}\mathbf{X}\mathbf{K}_{1}^{T},
	\end{equation}
	where $\mathbf{K}_{l},\;l=1,2$ are the orthonormal matrices by taking eigenvalue decomposition (\ref{chuwz:10}) on the fractional Laplacian $\mathbf{L}_{l}^{\alpha},\;l=1,2$. Note that $\mathbf{P}_{\boxtimes}=\mathbf{Q}_{\boxtimes}=\mathbf{K}_{1}\otimes\mathbf{K}_{2}$ in $(\ref{chuwz:9})$. Then, we can easily obtain
	\begin{equation*}
		\mathcal{F}_{\boxtimes}^\alpha[{\rm{vec}}({\mathbf{X}})]=\left(\begin{matrix}{\rm{vec}}(\mathcal{F}_{\alpha}\mathbf{X})\\
			\mathbf{0}\end{matrix}\right).
	\end{equation*}
	Therefore, for undirected graphs, our new GFRFT $\mathcal{F}_{\boxtimes}^\alpha$ in (\ref{chuwz:12}) is essentially consistent with Laplacian based multi-dimensional GFRFT in \cite{YL22}.
	
	\begin{Remark}
		When $\alpha=1$, the GFRFT $\mathcal{F}_{\boxtimes}^\alpha$ $(\ref{chuwz:12})$ reduces to GFT $\mathcal{F}_{\boxtimes}$ $(\ref{chuwz:4})$ mentioned in $\cite{CCL23}$. Hence, our GFRFT $\mathcal{F}_{\boxtimes}^\alpha$ is a natural extension from GFT domain to fractional order.
	\end{Remark}
	
	Motivated by \cite{CCL23}, we consider the singular values $r_i,0\le i\le N-1$ as frequencies of the GFRFT $\mathcal{F}_{\boxtimes}^\alpha$, and $\mathbf{p}_{k},\mathbf{q}_{k},\;0\le k\le N-1$, as the left and right frequency components, respectively. Then, we demonstrate that the energy of signals defined on a directed Cartesian product graph $\mathcal{G}$ with strong spatiotemporal correlation mainly concentrated in the low frequencies of GFRFT $\mathcal{F}_{\boxtimes}^\alpha$.
	
	\begin{Theorem}\label{Thm-GFRFT1}
		Assume that $\mathcal{G}=\mathcal{G}_{1}\boxtimes\mathcal{G}_{2}$ is a Cartesian product graph of two directed graphs $\mathcal{G}_{1}$ and $\mathcal{G}_{2}$, the fractional Laplacian matrix $\mathbf{L}_{\boxtimes}^{\alpha}$ on $\mathcal{G}$ is defined the same as $(\ref{def:frac})$, and $\mathbf{p}_i,\;\mathbf{q}_i,\;r_i,\;0\le i\le N-1$ are the same as in $(\ref{chuwz:9})$, $\alpha$ is the fractional order, which satisfies $0<\alpha\le 1$. Let $\Omega\in\{1,2,\cdots,N\}$ be the frequency bandwidth of the GFRFT $\mathcal{F}_{\boxtimes}^\alpha$ in $(\ref{chuwz:12})$, and the low frequency component of a signal $\mathbf{x}$ on $\mathcal{G}$ be
		\begin{eqnarray}\label{chuwz:17}
			\mathbf{x}_{\Omega,\boxtimes}^{\alpha}&:=&\frac{1}{2}\sum_{i=0}^{\Omega-1}[(y_{1,i}+y_{2,i})\mathbf{p}_{i}+(y_{1,i}-y_{2,i})\mathbf{q}_{i}]\nonumber\\
			&=&\frac{1}{2}\sum_{i=0}^{\Omega-1}(\mathbf{p}_{i}\mathbf{p}_{i}^{T}+\mathbf{q}_{i}\mathbf{q}_{i}^{T})\mathbf{x},
		\end{eqnarray}
		where
		\[
		y_{1,i}:=\frac{(\mathbf{p}_{i}+\mathbf{q}_{i})^{T}\mathbf{x}}{2},\;\; y_{2,i}:=\frac{(\mathbf{p}_{i}-\mathbf{q}_{i})^{T}\mathbf{x}}{2},
		\]
		for all $0\leq i\leq\Omega-1$.
		Then, we have
		\begin{align}
			\|\mathbf{x}-\mathbf{x}_{\Omega,\boxtimes}^{\alpha}\|_{2}\le&\frac{1}{2r_{\Omega-1}}(\|\mathbf{L}_{\boxtimes}^{\alpha}\mathbf{x}\|_{2}
			+\|(\mathbf{L}_{\boxtimes}^{\alpha})^{T}\mathbf{x}\|_{2})\nonumber\\
			\le&\frac{1}{2r_{\Omega-1}}\left[\|(\mathbf{L}_{1}^{\alpha}\otimes\mathbf{I}_{N_{2}})
			\mathbf{x}\|_{2}+\|((\mathbf{L}_{1}^{\alpha})^{T}\otimes\mathbf{I}_{N_{2}})\mathbf{x}\|_{2}\right.\nonumber\\
			&\quad \left.+\|(\mathbf{I}_{N_{1}}\otimes\mathbf{L}_{2}^{\alpha})\mathbf{x}\|_{2}
			+\|(\mathbf{I}_{N_{1}}\otimes(\mathbf{L}_{2}^{\alpha})^{T})\mathbf{x}\|_{2}\right],\label{chuwz:18}
		\end{align}
		where $r_{\Omega-1}$ is the cutoff frequency.
	\end{Theorem}
	
	\begin{proof}
		From (\ref{chuwz:9}), we get
		\begin{align}
			\|\mathbf{L}_{\boxtimes}^{\alpha}\mathbf{x}\|_{2}^{2}=\mathbf{x}^{T}\mathbf{Q}_{\boxtimes}\mathbf{R}^{2}\mathbf{Q}_{\boxtimes}^{T}\mathbf{x}
			=\sum_{i=0}^{N-1}r_{i}^{2}(\mathbf{q}_{i}^{T}\mathbf{x})^{2}
			\ge r_{\Omega-1}^{2}\sum_{i=\Omega}^{N-1}(\mathbf{q}_{i}^{T}\mathbf{x})^{2},\label{eq:GFRFT1}
		\end{align}
		and
		\begin{align}
			\|(\mathbf{L}_{\boxtimes}^{\alpha})^{T}\mathbf{x}\|_{2}^{2}
			=\mathbf{x}^{T}\mathbf{P}_{\boxtimes}\mathbf{R}^{2}\mathbf{P}_{\boxtimes}^{T}\mathbf{x}
			=\sum_{i=0}^{N-1}r_{i}^{2}(\mathbf{p}_{i}^{T}\mathbf{x})^{2}\ge r_{\Omega-1}^{2}\sum_{i=\Omega}^{N-1}(\mathbf{p}_{i}^{T}\mathbf{x})^{2}.\label{eq:GFRFT2}
		\end{align}
		
		Combining (\ref{chuwz:15}) and (\ref{chuwz:17}), yields
		\begin{align}
			\|\mathbf{x}-\mathbf{x}_{\Omega,\boxtimes}^{\alpha}\|_{2}
			=&\frac{1}{2}\left\|\sum_{i=\Omega}^{N-1}(\mathbf{p}_{i}\mathbf{p}_{i}^{T}+\mathbf{q}_{i}\mathbf{q}_{i}^{T})\mathbf{x}\right\|_{2}\nonumber\\ \le&\frac{1}{2}\left[\sum_{i=\Omega}^{N-1}(\mathbf{p}_{i}^{T}\mathbf{x})^{2}\right]^{1/2}
			+\frac{1}{2}\left[\sum_{i=\Omega}^{N-1}(\mathbf{q}_{i}^{T}\mathbf{x})^{2}\right]^{1/2}.\label{eq:GFRFT3}
		\end{align}
		
		Substituting (\ref{eq:GFRFT1}) and (\ref{eq:GFRFT2}) into (\ref{eq:GFRFT3}), we obtain
		\begin{eqnarray}
			\|\mathbf{x}-\mathbf{x}_{\Omega,\boxtimes}^{\alpha}\|_{2}\le\frac{1}{2r_{\Omega-1}}[\|\mathbf{L}_{\boxtimes}^{\alpha}\mathbf{x}\|_{2}
			+\|(\mathbf{L}_{\boxtimes}^{\alpha})^{T}\mathbf{x}\|_{2}].\label{eq:GFRFT4}
		\end{eqnarray}
		
		Combining (\ref{def:frac}) and (\ref{eq:GFRFT4}), we can get (\ref{chuwz:18}), which completes the proof.
	\end{proof}
	
	\section{Another SVD-Based GFRFT On a Directed Cartesian Product Graph}\label{sec4}
	
	Sometimes, some graph signals have different correlation characteristics in different directions, such as spatiotemporal signals. Therefore, defining GFRFT should reflect the spectral characteristics in different directions of graph signals. In this section, we propose a novel GFRFT $\mathcal{F}_{\otimes}^\alpha$ on a directed Cartesian product graph $\mathcal{G}=\mathcal{G}_{1}\boxtimes\mathcal{G}_{2}$, and show that $\mathcal{F}_{\otimes}^\alpha$ has lower computational complexity than $\mathcal{F}_{\boxtimes}^\alpha$. Moreover, we prove that $\mathcal{F}_{\otimes}^\alpha$ can  effectively represent graph signals that exhibit strong spatial-temporal correlation.
	
	First, we suppose that $\mathcal{G}:=\mathcal{G}_{1}\boxtimes\mathcal{G}_{2}$ is a Cartesian product graph of two directed graphs $\mathcal{G}_{1}=(\mathbf{V}_{1},\mathbf{E}_{1},\mathbf{A}_1)$ and $\mathcal{G}_{2}=(\mathbf{V}_{2}, \mathbf{E}_{2},\mathbf{A}_2)$. For $0<\alpha\le 1$, the graph fractional Laplacian matrices $\mathbf{L}_{l}^{\alpha},l=1,2$ are defined the same as in (\ref{eq:La}):
	\begin{equation}\label{eq:add1}
		\mathbf{L}^{\alpha}_{l}=\mathbf{P}_{l}\mathbf{R}_{l}\mathbf{Q}^{T}_{l}=\sum_{i=0}^{N_l-1}r_{l,i}\mathbf{p}_{l,i}\mathbf{q}_{l,i}.
	\end{equation}
	
	Let
	\begin{equation}\label{eq:add3}
		\mathbf{P}_{\otimes}=\mathbf{P}_{1}\otimes\mathbf{P}_{2},\;\; \mathbf{Q}_{\otimes}=\mathbf{Q}_{1}\otimes \mathbf{Q}_{2}.
	\end{equation}
	Then, based on $\mathbf{P}_{\otimes}$ and $\mathbf{Q}_{\otimes}$, we propose another GFRFT defined on a directed Cartesian product graph $\mathcal{G}$.
	
	\begin{Definition}\label{Def:GFRFT2}
		Let $\mathcal{G}:=\mathcal{G}_{1}\boxtimes\mathcal{G}_{2}$ be a Cartesian product graph of two directed graphs $\mathcal{G}_{l}$ and $\mathcal{G}_{2}$, and fractional Laplacian matrices $\mathbf{L}^{\alpha}_{l}$ be given by $(\ref{eq:add1})$, $\mathbf{P}_{\otimes}$ and $\mathbf{Q}_{\otimes}$ be defined as $(\ref{eq:add3})$. Then, the
		GFRFT $\mathcal{F}_{\otimes}^\alpha:\mathbb{R}^{N}\mapsto \mathbb{R}^{2N}$ of a signal $\mathbf{x}\in\mathbb{R}^{N}$ on a directed Cartesian product graph $\mathcal{G}$ is defined by
		\begin{equation}\label{chuwz:21}
			\mathcal{F}_{\otimes}^\alpha\mathbf{x}:=\dfrac{1}{2}\left(\begin{matrix}(\mathbf{P}_{\otimes}+\mathbf{Q}_{\otimes})^{T}\mathbf{x}\\
				(\mathbf{P}_{\otimes}-\mathbf{Q}_{\otimes})^{T}\mathbf{x} \end{matrix}\right).		
		\end{equation}
		Moreover, the inverse GFRFT $\mathcal{F}_{\otimes}^{-\alpha}: \mathbb{R}^{2N}\mapsto \mathbb{R}^{N}$ is given by
		\begin{equation}\label{chuwz:22}
			\mathcal{F}_{\otimes}^{-\alpha}
			\left(\begin{matrix}\mathbf{y}_{1}\\
				\mathbf{y}_{2}\end{matrix}\right)
			:=\frac{1}{2}[\mathbf{P}_{\otimes}(\mathbf{y}_{1}+\mathbf{y}_{2})+\mathbf{Q}_{\otimes}(\mathbf{y}_{1}-\mathbf{y}_{2})],
		\end{equation}
		where $\mathbf{y}_{1},\mathbf{y}_{2}\in\mathbb{R}^{N}$.
	\end{Definition}
	
	For the new GFRFT $\mathcal{F}_{\otimes}^\alpha$, we consider singular value pairs $(r_{1,i},r_{2,j})$ of fractional Laplacian matrices $\mathbf{L}_{1}^{\alpha}$ and $\mathbf{L}_{2}^{\alpha}$ as frequency pairs of the GFRFT, $\mathbf{p}_{1,i}\otimes\mathbf{p}_{2,j}$ and $\mathbf{q}_{1,i}\otimes\mathbf{q}_{2,j}\;(0\leq i\leq N_{1}-1,0\leq j\leq N_{2}-1)$ as the left and right frequency components, respectively. The computational complexity for calculating the left or right frequency components of the GFRFT $\mathcal{F}_{\otimes}^\alpha$ is $\mathcal{O}(N_{1}^{3}+N_{2}^{3})$.	The parameter $\alpha$ makes our GFRFT more flexible without extra computational cost.
	
	\begin{Remark}
		Let $\mathbf{L}_{1,q}^\alpha$ and $\mathbf{L}_{2,q}^\alpha$,\;$q>0$, be the Hermitian fractional Laplacian matrices on the directed graphs $\mathcal{G}_{1}$ and $\mathcal{G}_{2}$, respectively. By performing SVDs on $\mathbf{L}_{l,q}^\alpha$,\;$l=1,2$, we have
		\[
		\mathbf{L}_{l,q}^{\alpha}=\mathbf{P}_{l,q}\mathbf{\Upsilon}_{l,q}\mathbf{P}_{l,q}^*,\;\;l=1,2,
		\]
		where
		\begin{equation*}
			\mathbf{P}_{l,q}=[\mathbf{p}_{l,q,0},\mathbf{p}_{l,q,1},\cdots,\mathbf{p}_{l,q,N_l-1}],\;\;
			\boldsymbol{\Upsilon}_{l,q}={\rm{diag}}([\mathbf{\varphi}_{l,q,0},\mathbf{\varphi}_{l,q,1},\cdots,\mathbf{\varphi}_{l,q,N_l-1}]).
		\end{equation*}
		Utilizing the argument mentioned in $(\ref{chuwz:11})$, we can represent the Hermitian fractional Laplacian matrix $\mathbf{L}_{\boxtimes,q}^\alpha$ on a directed Cartesian product graph $\mathcal{G}$ as
		\begin{eqnarray}
			\mathbf{L}_{\boxtimes,q}^{\alpha}=\sum^{N_{1}-1}_{i=0}\sum^{N_{2}-1}_{j=0}(\mathbf{\varphi}_{1,q,i}
			+\mathbf{\varphi}_{2,q,j})(\mathbf{p}_{1,q,i}\otimes \mathbf{p}_{2,q,j})
			\times(\mathbf{p}_{1,q,i}\otimes \mathbf{p}_{2,q,j})^{*}.\label{chuwz:38}
		\end{eqnarray}
		Hence, the computation complexity of performing SVD on the $\mathbf{L}_{\boxtimes,q}^{\alpha}$ is $\mathcal{O}(N_1^3+N_2^3)$. From $(\ref{Def:DGFRFT2})$, the DGFRFT on $\mathcal{G}$ is defined by
		\begin{equation}\label{Def:spqd}
			\mathcal{F}_{q}^\alpha\mathbf{x}=(\mathbf{P}_{1,q}\otimes\mathbf{P}_{2,q})^*\mathbf{x},
		\end{equation}
		which is coincides with $\mathcal{F}_{\boxtimes,q}^\alpha$ and $\mathcal{F}_{\otimes,q}^\alpha$.
	\end{Remark}
	
	For a signal $\mathbf{x}\in \mathbb{R}^N$ on a directed Cartesian product graph $\mathcal{G}$, it is easily to obtain that
	\begin{equation}\label{chuwz:23}
		\mathcal{F}_{\otimes}^{-\alpha}[\mathcal{F}_{\otimes}^\alpha\mathbf{x}]=\mathbf{x},
	\end{equation}	
	and
	\begin{equation}\label{chuwz:24}	
		\|\mathcal{F}_{\otimes}^\alpha\mathbf{x}\|_{2}=\|\mathbf{x}\|_{2}.
	\end{equation}	
	
	\begin{Remark}
		When $\alpha=1$, the GFRFT $\mathcal{F}_{\otimes}^\alpha$ $(\ref{chuwz:21})$ reduces to GFT $\mathcal{F}_{\otimes}$ $(\ref{chuwz:7})$ proposed in $\cite{CCL23}$. Hence, our GFRFT $\mathcal{F}_{\otimes}^\alpha$ is a generalization of GFT to fractional order.
	\end{Remark}
	
	In the following, we show that the energy of spatial-temporal signals on a directed Cartesian product graph $\mathcal{G}$ with strong correlation mainly concentrated in the low frequencies of the new GFRFT $\mathcal{F}_{\otimes}^\alpha$.
	
	\begin{Theorem}\label{Thm:An2}
		Suppose that $\mathcal{G}:=\mathcal{G}_{1}\boxtimes\mathcal{G}_{2}$ is a Cartesian product graph of two directed graphs $\mathcal{G}_{l}$ and $\mathcal{G}_{2}$, and fractional Laplacian matrices $\mathbf{L}^{\alpha}_{l}$ is given by $(\ref{eq:add1})$, $r_{l,i},\;\mathbf{p}_{l,i},\;\mathbf{q}_{l,i},0\leq i\le N_{l}-1,\;l=1,2$ are the same as in $(\ref{eq:add1})$, $\tau_k,\;0\le k\le N-1$ are a non-descending rearrangement of $r_{1,i}+r_{2,j}$,\;$0\leq i\leq N_{1}-1,\;0\leq j\leq N_{2}-1$. Let $\Omega\in[1,2,\cdots,N]$ be the frequency bandwidth of GFRFT $\mathcal{F}_{\otimes}^\alpha$ in $(\ref{chuwz:21})$, and the low frequency component of a signal $\mathbf{x}$ on $\mathcal{G}$ be
		\begin{eqnarray}\label{chuwz:25}
			\mathbf{x}_{\Omega,\otimes}^{\alpha}
			&=&\frac{1}{2}\sum_{(i,j)\in \mathcal{S}_{\Omega}}\Big[(\mathbf{p}_{1,i}\otimes\mathbf{p}_{2,j})(\mathbf{p}_{1,i}\otimes\mathbf{p}_{2,j})^{T}\mathbf{x}{}\nonumber\\
			&&{}+(\mathbf{q}_{1,i}\otimes\mathbf{q}_{2,j})(\mathbf{q}_{1,i}\otimes\mathbf{q}_{2,j})^{T}\mathbf{x}\Big],
		\end{eqnarray}
		where $\mathcal{S}_{\Omega}=\{(i,j)|\tau_k=r_{1,i}+r_{2,j},0\le k\le\Omega-1\}$. Then, we get
		\begin{align}\label{chuwz:26}
			\|\mathbf{x}-\mathbf{x}_{\Omega,\otimes}^{\alpha}\|_{2}
			\le&\frac{1}{2\tau_{\Omega-1}}\left[\|(\mathbf{L}_{1}^{\alpha}\otimes\mathbf{I}_{N_{2}})\mathbf{x}\|_{2}
			+\|((\mathbf{L}_{1}^{\alpha})^{T}\otimes\mathbf{I}_{N_{2}})\mathbf{x}\|_{2}\right.\nonumber\\
			&\quad\left.+\|(\mathbf{I}_{N_{1}}\otimes\mathbf{L}_{2}^{\alpha})\mathbf{x}\|_{2}+\|(\mathbf{I}_{N_{1}}\otimes(\mathbf{L}_{2}^{\alpha})^{T})\mathbf{x}\|_{2}\right],
		\end{align}	
		where $\tau_{\Omega-1}$ is the cut-off frequency.
	\end{Theorem}
	
	\begin{proof}
		From (\ref{eq:add1}), we obtain
		\begin{equation*}
			\|(\mathbf{L}_{1}^{\alpha}\otimes\mathbf{I}_{N_{2}})\mathbf{x}\|_{2}^{2}
			=\sum_{i=0}^{N_{1}-1}\sum_{j=0}^{N_{2}-1}r_{1,i}^{2}((\mathbf{q}_{1,i}\otimes\mathbf{q}_{2,j})^{T}\mathbf{x})^{2}
		\end{equation*}
		and
		\begin{equation*}
			\|(\mathbf{I}_{N_{1}}\otimes\mathbf{L}_{2}^{\alpha})\mathbf{x}\|_{2}^{2}
			=\sum_{i=0}^{N_{1}-1}\sum_{j=0}^{N_{2}-1}r_{2,j}^{2}((\mathbf{q}_{1,i}\otimes\mathbf{q}_{2,j})^{T}\mathbf{x})^{2}.
		\end{equation*}
		Therefore,
		\begin{align}		
         &(\|(\mathbf{L}_{1}^{\alpha}\otimes\mathbf{I}_{N_{2}})\mathbf{x}\|_{2}+\|(\mathbf{I}_{N_{1}}\otimes\mathbf{L}_{2}^{\alpha})\mathbf{x}\|_{2})^{2}\nonumber\\ \ge&\sum_{i=0}^{N_{1}-1}\sum_{j=0}^{N_{2}-1}(r_{1,i}+r_{2,j})^{2}((\mathbf{q}_{1,i}\otimes\mathbf{q}_{2,j})^{T}\mathbf{x})^{2}\nonumber\\
			\ge&\tau_{\Omega-1}^{2}\sum_{(i,j)\notin\mathcal{S}_{\Omega}}((\mathbf{q}_{1,i}\otimes\mathbf{q}_{2,j})^{T}\mathbf{x})^{2}.\label{eq:th1}
		\end{align}
		
		Similarly, it is follows from $(\ref{eq:add1})$ that
		\begin{align}
			&\left(\|((\mathbf{L}_{1}^{\alpha})^{T}\otimes \mathbf{I}_{N_{2}})\mathbf{x}\|_{2}+\|(\mathbf{I}_{N_{1}}\otimes (\mathbf{L}_{2}^{\alpha})^{T} )\mathbf{x}\|_{2}\right)^{2}\nonumber\\
			\ge&\tau_{\Omega-1}^{2}\sum_{(i,j)\notin \mathcal{S}_{\Omega}}((\mathbf{p}_{1,i}\otimes\mathbf{p}_{2,j})^{T}\mathbf{x})^{2}.\label{eq:th2}
		\end{align}
		
		Combining (\ref{chuwz:23}) and (\ref{chuwz:25}), we get
		\begin{align}
		&\|\mathbf{x}-\mathbf{x}_{\Omega,\otimes}^{\alpha}\|_{2}\nonumber\\
		=&\frac{1}{2}\Bigg\|\sum_{(i,j)\notin \mathcal{S}_{\Omega}}\Big[(\mathbf{p}_{1,i}\otimes\mathbf{p}_{2,j})(\mathbf{p}_{1,i}\otimes\mathbf{p}_{2,j})^{T}\mathbf{x}+(\mathbf{q}_{1,i}\otimes\mathbf{q}_{2,j})(\mathbf{q}_{1,i}\otimes\mathbf{q}_{2,j})^{T}\mathbf{x}\Big]\Bigg\|_2\nonumber\\
		\le&\frac{1}{2}\left[\sum_{(i,j)\notin \mathcal{S}_{\Omega}}((\mathbf{p}_{1,i}\otimes\mathbf{p}_{2,j})^{T}\mathbf{x})^{2}\right]^{1/2}\nonumber\\
         &\quad+\frac{1}{2}\left[\sum_{(i,j)\notin\mathcal{S}_{\Omega}}((\mathbf{q}_{1,i}\otimes\mathbf{q}_{2,j})^{T}\mathbf{x})^{2}\right]^{1/2}.\label{eq:th3}
	\end{align}
		
		Substituting (\ref{eq:th1}) and (\ref{eq:th2}) into (\ref{eq:th3}), we have
		\begin{align*}
			\|\mathbf{x}-\mathbf{x}_{\Omega,\otimes}^{\alpha}\|_{2}
			\le&\frac{1}{2\tau_{\Omega-1}}\left[\|(\mathbf{L}_{1}^{\alpha}\otimes\mathbf{I}_{N_{2}})\mathbf{x}\|_{2}
			+\|((\mathbf{L}_{1}^{\alpha})^{T}\otimes\mathbf{I}_{N_{2}})\mathbf{x}\|_{2}\right.\nonumber\\
			&\quad\quad\left.+\|(\mathbf{I}_{N_{1}}\otimes\mathbf{L}_{2}^{\alpha})\mathbf{x}\|_{2}+\|(\mathbf{I}_{N_{1}}\otimes(\mathbf{L}_{2}^{\alpha})^{T})\mathbf{x}\|_{2}\right],
		\end{align*}
		which completes the proof.
	\end{proof}
	
	For a graph signal $\mathbf{X}\in \mathbb{R}^{N_2\times N_1}$ on directed Cartesian product graph $\mathcal{G}$, the
	GFRFT $\mathcal{F}_{\otimes}^\alpha$ of $\mathbf{X}$ can be rewritten as
	\begin{equation}\label{chuwz:27} \mathcal{F}_{\otimes}^\alpha{\rm{vec}}(\mathbf{X})=\frac{1}{2}\left(\begin{matrix}{\rm{vec}}(\mathbf{P}_{2}^{T}\mathbf{X}\mathbf{P}_{1}+\mathbf{Q}_{2}^{T}\mathbf{X}\mathbf{Q}_{1})\\
			{\rm{vec}}(\mathbf{P}_{2}^{T}\mathbf{X}\mathbf{P}_{1}-\mathbf{Q}_{2}^{T}\mathbf{X}\mathbf{Q}_{1}) \end{matrix}\right).
	\end{equation}	
	
	Then, we can first obtain GFRFT $\mathcal{F}_{\otimes}^\alpha$ in the direction of the graph $\mathcal{G}_1$, then in the direction of $\mathcal{G}_{2}$ (see Algorithm \ref{algo1}).
	
	\begin{algorithm}[!h]
		\caption{Algorithm to perform the GFRFT $\mathcal{F}_{\otimes}^\alpha$}
		\vspace{1mm}
		\textbf{Input}: A graph signal $\mathbf{X}$.\\
		\textbf{Steps}:\\
		\vspace{-5mm}
		\begin{algorithmic}[1]\label{algo1}
			\STATE Do $\mathbf{Z}_{1}=\mathbf{X}\mathbf{P}_{1}$ and $\widetilde{\mathbf{Z}}_{1}=\mathbf{X}\mathbf{Q}_{1}$;
			\STATE Do $\mathbf{Z}_{2}=\mathbf{P}_{2}^{T}\mathbf{Z}_1$ and $\widetilde{\mathbf{Z}}_{2}=\mathbf{Q}_{2}^{T}\widetilde{\mathbf{Z}}_{1}$;
			\STATE Do $\mathcal{F}_{\otimes}^\alpha \mathbf{X}_{1}=\frac{\mathbf{Z}_{2}+\widetilde{\mathbf{Z}}_{2}}{2}$ and $\mathcal{F}_{\otimes}^\alpha \mathbf{X}_{2}=\frac{\mathbf{Z}_{2}-\widetilde{\mathbf{Z}}_{2}}{2}$.
		\end{algorithmic}
		\textbf{Outputs}: $\mathcal{F}_{\otimes}^\alpha \mathbf{X}_{1}$ and $\mathcal{F}_{\otimes}^\alpha \mathbf{X}_{2}$ are two components of the GFRFT
		$\mathcal{F}_{\otimes}^\alpha{\rm{vec}}(\mathbf{X})$.
	\end{algorithm}
	
	Similarly, the inverse GFRFT $\mathcal{F}_{\otimes}^{-\alpha}$ can be represented by
	\begin{equation*}
		\mathcal{F}_{\otimes}^{-\alpha}\left(\begin{matrix}\mathbf{y}_{1}\\
			\mathbf{y}_{2}\end{matrix}\right)
		=\frac{1}{2}(\mathbf{P}_{2}(\mathbf{Y}_{1}+\mathbf{Y}_{2})\mathbf{P}_{1}^{T}+\mathbf{Q}_{2}(\mathbf{Y}_{1}-\mathbf{Y}_{2})\mathbf{Q}_{1}^{T}),
	\end{equation*}	
	for all $\mathbf{y}_1,\mathbf{y}_2\in \mathbb{R}^N$, where $\mathbf{Y}_{i}={\rm{vec}}^{-1}(\mathbf{y}_i),\; i=1,2$. Then we can obtain the original signal $\mathbf{x}$ by Algorithm \ref{algo2}.
	
	\begin{algorithm}[!h]
		\caption{Algorithm to perform the Inverse GFRFT $\mathcal{F}_{\otimes}^{-\alpha}$}
		\vspace{1mm}
		\textbf{Inputs vectorization}: $\mathbf{Y}_{1}={\rm{vec}}^{-1}(\mathbf{y}_{1})$ and $\mathbf{Y}_{2}={\rm{vec}}^{-1}(\mathbf{y}_{2})$.\\
		\textbf{Steps}:\\
		\vspace{-5mm}
		\begin{algorithmic}[1]\label{algo2}
			\STATE Do $\mathbf{W}_{1}=(\mathbf{Y}_{1}+\mathbf{Y}_{2})\mathbf{P}_{1}^{T}$ and $\widetilde{\mathbf{W}}_{1}=(\mathbf{Y}_{1}-\mathbf{Y}_{2})\mathbf{Q}_{1}^{T}$;
			\STATE Do $\mathbf{W}_{2}=\mathbf{P}_{2}\mathbf{W}_{1}$ and $\widetilde{\mathbf{W}}_{2}=\mathbf{Q}_{2}\widetilde{\mathbf{W}}_{1}$;
			\STATE Do $\mathbf{X}$=$\frac{\mathbf{W}_{2}+\widetilde{\mathbf{W}}_{2}}{2}$.
		\end{algorithmic}
		\textbf{Outputs}: $\mathbf{x}={\rm{vec}}(\mathbf{X})=\mathcal{F}_{\otimes}^{-\alpha}\left(\begin{matrix}\mathbf{y}_{1}\\
			\mathbf{y}_{2}\end{matrix}\right)$.
	\end{algorithm}
	
	If $\mathcal{G}$ is an undirected Cartesian product graph, then $\mathbf{P}_{\boxtimes}=\mathbf{Q}_{\boxtimes}$ in $(\ref{chuwz:9})$ and $\mathbf{P}_{\otimes}=\mathbf{Q}_{\otimes}$ in $(\ref{eq:add3})$ are equal. Hence, for any graph signal $\mathbf{x}$ on $\mathcal{G}$, we have
	\begin{equation}\label{chuwz:28}
		\mathcal{F}_{\boxtimes}^\alpha\mathbf{x}=\mathcal{F}_{\otimes}^\alpha\mathbf{x}
		=\left(\begin{matrix}\mathbf{P}_{\boxtimes}^{T}{\mathbf{x}}\\
			\mathbf{0}_N\end{matrix}\right).
	\end{equation}	
	
	When $\alpha=1$, Cheng et al \cite{CCL23} proved that two GFTs $\mathcal{F}_{\boxtimes}$ and $\mathcal{F}_{\otimes}$ are identical only for the undirected Cartesian product graph $\mathcal{G}$. Therefore, our two GFRFTs are not the same in general.
	
\begin{Remark}
Unlike undirected graphs, the fractional Laplacian matrix for directed graphs is asymmetric, thus eigenvalue decomposition is ineffective. $\mathcal{F}_{\boxtimes}^{\alpha}$ and $\mathcal{F}_{\otimes}^{\alpha}$ address this problem using singular value decomposition (SVD). Previous work has primarily focused on undirected product graphs \cite{YL22}, or single directed graphs \cite{YL23}, and this paper extends the SVD-based GFT method \cite{CCL23} to fractional domain on directed Cartesian product graphs, a structure crucial for modeling spatiotemporal data, such as weather station time series.
Moreover, $\mathcal{F}_{\boxtimes}^{\alpha}$ is suitable for global filtering of coupled signals, a low-pass filter that preserves a small $r_i$ can suppress spatiotemporal synchronization noise. While the transform $\mathcal{F}_{\otimes}^{\alpha}$ is used to separate signals, a filter that preserves smaller $r_{1,i}$ (low-frequency time) and larger $r_{2,j}$ (high-frequency space) can retain slow trends while enhancing spatial details.
Furthermore, $\mathcal{F}_{\boxtimes}^{\alpha}$ requires $\mathcal{O}((N_1N_2)^3)$ operations for SVD of the fractional Laplacian matrix $L_{\boxtimes}^\alpha$, restricting its use to small-scale directed product graphs. Conversely, $\mathcal{F}_{\otimes}^{\alpha}$ achieves $\mathcal{O}(N_1^3+N_2^3)$ complexity through separable subgraph processing, making it feasible for large-scale directed product graphs. Finally, the fractional order $q$ of our proposed GFRFTs can be freely chosen, making them more flexible and practical than GFTs mentioned in \cite{CCL23}.
\end{Remark}

	\section{SVD-Based MGFRFT on a Cartesian Product of m Directed Graphs}\label{sec5}
	
	A signal defined on a Cartesian product graph can be viewed as a two-dimensional signal, but it is not real two-dimensional data. Instead, it is a graph signal distributed across two-dimensional graph nodes. Therefore, this concept can be naturally extended to multidimensional graph signals. For any natural number $m$, an $m$-dimensional GFRFT defined on the Cartesian product graph ${\mathcal{G}_1}\boxtimes\cdots\boxtimes{\mathcal{G}_m}=((\cdots({\mathcal{G}_1}\boxtimes{\mathcal{G}_2})\cdots)\boxtimes{\mathcal{G}_m})$ is obtained inductively from the definition of the 2D GFRFT. In this section, we extend the definitions of GFRFTs on a directed Cartesian product graph from two graphs to $m$ graphs setting.
	
	First, we consider a directed Cartesian product graph $\mathcal{G}=\mathcal{G}_{1}\boxtimes\mathcal{G}_{2}\boxtimes\cdots\boxtimes\mathcal{G}_{m}$, where $\mathcal{G}_{i}=(\mathbf{V}_{i},\mathbf{E}_{i},\mathbf{A}_i)$, $i=1,2,\cdots,m$ are directed graphs. By performing SVD, the Laplacian matrices $\mathbf{L}_l,\;l=1,2,\cdots,m$ of graph $\mathcal{G}_l$ can be decomposed into
	\[
	\mathbf{L}_{l}=\mathbf{U}_{l}\mathbf{\Sigma}_{l}\mathbf{V}_{l}^{T},\; l=1,2,\cdots,m,
	\]
	and the fractional Laplacian matrices $\mathbf{L}_{l}^{\alpha},l=1,2,\cdots,m$ are defined the same as (\ref{eq:La}):
	\begin{equation}\label{eq:mLa}
		\mathbf{L}^{\alpha}_{l}=\mathbf{P}_{l}\mathbf{R}_{l}\mathbf{Q}^{T}_{l}=\sum_{i=0}^{N_l-1}r_{l,i}\mathbf{p}_{l,i}\mathbf{q}_{l,i},\;\;l=1,2,\cdots,m.
	\end{equation}
	
Then, the fractional Laplacian matrix $\mathbf{L}_{m,\boxtimes}^{\alpha}$ for a directed Cartesian product graph $\mathcal{G}$ of $m$ directed graphs is given by
	\begin{align}\label{eq:mfrac}
		\mathbf{L}_{m,\boxtimes}^{\alpha}:=&\mathbf{L}_1^{\alpha}\oplus\mathbf{L}_2^{\alpha}\oplus\cdots\oplus\mathbf{L}_m^{\alpha}\nonumber\\
		=&\sum_{i=1}^{m}\mathbf{I}_{N_1N_2\cdots N_{i-1}}\otimes\mathbf{L}_i^{\alpha}\otimes \mathbf{I}_{N_{i+1}N_{i+2}\cdots N_m}.
	\end{align}
	
	By performing SVD on $\mathbf{L}_{m,\boxtimes}^{\alpha}$, it can be represented by
	\begin{equation}\label{eq:mfrc1}
		\mathbf{L}_{m,\boxtimes}^{\alpha}=\mathbf{P}_{m,\boxtimes}\mathbf{R}_{m,\boxtimes}\mathbf{Q}_{m,\boxtimes}^{T}
		=\sum_{k=0}^{N-1}r_{m,k}\mathbf{p}_{m,k}\mathbf{q}_{m,k}^{T},
	\end{equation}
	where $N=N_1N_2\cdots N_m$, matrices
	\[
	\mathbf{P}_{m,\boxtimes}=[\mathbf{p}_{m,0},\mathbf{p}_{m,1},\cdots,\mathbf{p}_{m,N-1}],\;\;
	\mathbf{Q}_{m,\boxtimes}=[\mathbf{q}_{m,0},\mathbf{q}_{m,1},\cdots,\mathbf{q}_{m,N-1}]
	\]
	are orthonormal,
	\begin{equation*}
		\mathbf{R}_{m,\boxtimes}={\rm{diag}}([r_{m,0},r_{m,1},\cdots,r_{m,N-1}]),
	\end{equation*}
	which satisfies $0=r_{m,0}\le r_{m,1}\le\cdots\le r_{m,N-1}$.
	The time complexity for computing the SVD factorization of $\mathbf{L}_{m,\boxtimes}^{\alpha}$ is $\mathcal{O}(N^{3})$.
	
	Next, based on the SVD of $\mathbf{L}_{m,\boxtimes}^{\alpha}$, we define the GFRFT of a graph signal on a Cartesian product graph with $m$ directed graphs (MGFRFT).
	
	\begin{Definition}\label{Def-mFT}
		Suppose that $\mathcal{G}=\mathcal{G}_{1}\boxtimes\mathcal{G}_{2}\boxtimes\cdots\boxtimes\mathcal{G}_{m}$ is a Cartesian product of $m$ directed graphs $\mathcal{G}_{l},l=1,2,\cdots,m$, the fractional Laplacian matrix $\mathbf{L}_{m,\boxtimes}^{\alpha}$ on $\mathcal{G}$ is defined the same as $(\ref{eq:mfrac})$, and has the SVD form as in $(\ref{eq:mfrc1})$, $\alpha$ is the fractional order, which satisfies $0<\alpha\le 1$. The MGFRFT $\mathcal{F}_{m,\boxtimes}^\alpha$ of a signal $\mathbf{x}:\mathbf{V}_1\times\mathbf{V}_2\times\cdots\times\mathbf{V}_m\rightarrow \mathbb{R}^{N}$ on $\mathcal{G}$ is defined by
		
		\begin{eqnarray}\label{m1}
			\mathcal{F}_{m,\boxtimes}^\alpha\mathbf{x} :=\dfrac{1}{2}\left(\begin{matrix}(\mathbf{P}_{m,\boxtimes}+\mathbf{Q}_{m,\boxtimes})^{T}\mathbf{x}\\
				(\mathbf{P}_{m,\boxtimes}-\mathbf{Q}_{m,\boxtimes})^{T}\mathbf{x} \end{matrix}\right)
			=\frac{1}{2}\left(\begin{matrix}(\mathbf{p}_{m,0}+\mathbf{q}_{m,0})^{T}\mathbf{x}\\
				\vdots\\
				(\mathbf{p}_{m,N-1}+\mathbf{q}_{m,N-1})^{T}\mathbf{x}\\
				(\mathbf{p}_{m,0}-\mathbf{q}_{m,0})^{T}\mathbf{x}\\
				\vdots\\
				(\mathbf{p}_{m,N-1}-\mathbf{q}_{m,N-1})^{T}\mathbf{x}\end{matrix}\right).
		\end{eqnarray}
		
		In addition, the inverse MGFRFT $\mathcal{F}_{m,\boxtimes}^{-\alpha}$ is defined as
		\begin{align}
			\mathcal{F}_{m,\boxtimes}^{-\alpha}\left(\begin{matrix}\mathbf{y}_{1}\\ \mathbf{y}_{2}\end{matrix}\right)		:=&\frac{1}{2}\left[\mathbf{P}_{m,\boxtimes}(\mathbf{y}_{1}+\mathbf{y}_{2})+\mathbf{Q}_{m,\boxtimes}(\mathbf{y}_{1}-\mathbf{y}_{2})\right]\nonumber\\
			=&\frac{1}{2}\sum_{i=0}^{N-1}\left[(y_{1,i}+y_{2,i})\mathbf{p}_{m,i}+(y_{1,i}-y_{2,i})\mathbf{q}_{m,i}\right],\label{m2}
		\end{align}
		for all $\mathbf{y}_l=[y_{l,0},y_{l,1},\cdots,y_{l, N-1}]^{T}\in \mathbb{R}^N,\;l=1,2.$
	\end{Definition}
	
%
	Next, we show that most of the energy of a graph signal on $\mathcal{G}$ with strong spatiotemporal correlation is concentrated in the low frequencies of MGFRFT $\mathcal{F}_{m,\boxtimes}^\alpha$.
	
	\begin{Theorem}\label{Thm-mGF}
		Suppose that $\mathcal{G}=\mathcal{G}_{1}\boxtimes\mathcal{G}_{2}\boxtimes\cdots\boxtimes\mathcal{G}_{m}$ is a Cartesian product of $m$ directed graphs $\mathcal{G}_{l},l=1,2,\cdots,m$, the fractional Laplacian matrix $\mathbf{L}_{m,\boxtimes}^{\alpha}$ on $\mathcal{G}$ is defined the same as $(\ref{eq:mfrac})$, and $\mathbf{p}_{m,i},\;\mathbf{q}_{m,i},\;r_{m,i},\;0\le i\le N-1$ are the same as in $(\ref{eq:mfrc1})$, $\alpha$ is the fractional order, which satisfies $0<\alpha\le 1$. Let $\Gamma\in\{1,2,\cdots,N\}$ be the frequency bandwidth of the MGFRFT $\mathcal{F}_{m,\boxtimes}^\alpha$ in $(\ref{m1})$, and the low frequency component of a signal $\mathbf{x}$ on $\mathcal{G}$ be
		\begin{eqnarray}
			\mathbf{x}_{\Gamma,m,\boxtimes}^{\alpha}&:=&\frac{1}{2}\sum_{i=0}^{\Gamma-1}[(y_{1,i}+y_{2,i})\mathbf{p}_{m,i}+(y_{1,i}-y_{2,i})\mathbf{q}_{m,i}]\nonumber\\
			&=&\frac{1}{2}\sum_{i=0}^{\Gamma-1}(\mathbf{p}_{m,i}\mathbf{p}_{m,i}^{T}+\mathbf{q}_{m,i}\mathbf{q}_{m,i}^{T})\mathbf{x},\label{Thm-mGF1}
		\end{eqnarray}
		where
		\[
		y_{1,i}:=\frac{(\mathbf{p}_{m,i}+\mathbf{q}_{m,i})^{T}\mathbf{x}}{2},\; y_{2,i}:=\frac{(\mathbf{p}_{m,i}-\mathbf{q}_{m,i})^{T}\mathbf{x}}{2},
		\]
		for all $0\leq i\leq\Gamma-1.$
		Then, we have
		\begin{align*}
			&\|\mathbf{x}-\mathbf{x}_{\Gamma,m,\boxtimes}^{\alpha}\|_{2}
			\le\frac{1}{2r_{m,\Gamma-1}}(\|\mathbf{L}_{m,\boxtimes}^{\alpha}\mathbf{x}\|_{2}
			+\|(\mathbf{L}_{m,\boxtimes}^{\alpha})^{T}\mathbf{x}\|_{2})\\
			\le&\frac{1}{2r_{m,\Gamma-1}}\Big[\|(\mathbf{L}_1^{\alpha}\otimes \mathbf{I}_{N_2N_{3}\cdots N_m})
			\mathbf{x}\|_{2}+\|(\mathbf{I}_{N_1}\otimes\mathbf{L}_2^{\alpha}\otimes \mathbf{I}_{N_{3}N_{4}\cdots N_m})\mathbf{x}\|_{2}\\
			&\quad +\cdots +\|(\mathbf{I}_{N_{1}N_{2}\cdots N_{m-1}}\otimes\mathbf{L}_m^{\alpha})\mathbf{x}\|_{2}
			+\|((\mathbf{L}_{1}^{\alpha})^{T}\otimes\mathbf{I}_{N_2N_{3}\cdots N_m})\mathbf{x}\|_{2}.\\
			&\quad+\|(\mathbf{I}_{N_1}\otimes(\mathbf{L}_{2}^{\alpha})^{T}\otimes \mathbf{I}_{N_{3}N_{4}\cdots N_m})\mathbf{x}\|_{2}
		+\cdots +\|(\mathbf{I}_{N_{1}N_{2}\cdots N_{m-1}}\otimes(\mathbf{L}_{m}^{\alpha})^{T})\mathbf{x}\|_{2}\Big],
		\end{align*}
		where $r_{m,\Gamma-1}$ is the cutoff frequency.
	\end{Theorem}
	
	Set
	\begin{eqnarray}\label{Thm-mGF2}
		\mathbf{P}_{m,\otimes}=\mathbf{P}_{1}\otimes\mathbf{P}_{2}\otimes\cdots\otimes\mathbf{P}_{m},\;\;
		\mathbf{Q}_{m,\otimes}=\mathbf{Q}_{1}\otimes\mathbf{Q}_{2}\otimes\cdots\otimes\mathbf{Q}_{m}.
	\end{eqnarray}	
	Then, we define another MGFRFT on $\mathcal{G}$ by $\mathbf{P}_{m,\otimes}$ and $\mathbf{Q}_{m,\otimes}$ as follows.
	
	\begin{Definition}\label{Def:mGF2}
		Let $\mathcal{G}=\mathcal{G}_{1}\boxtimes\mathcal{G}_{2}\boxtimes\cdots\boxtimes\mathcal{G}_{m}$ be a Cartesian product graph of $m$ directed graphs $\mathcal{G}_{l},\;l=1,2,\cdots,m$, and fractional Laplacian matrices $\mathbf{L}_{l}^{\alpha},\;\;l=1,2,\cdots,m$ be given by $(\ref{eq:mLa})$, $\mathbf{P}_{m,\otimes}$ and $\mathbf{Q}_{m,\otimes}$ be defined as $(\ref{Thm-mGF2})$. Then, the
		MGFRFT $\mathcal{F}_{m,\otimes}^\alpha:\mathbb{R}^{N}\mapsto \mathbb{R}^{2N}$ of a signal $\mathbf{x}\in\mathbb{R}^{N}$ on the directed Caresian product graph $\mathcal{G}$ is defined as
		\begin{equation}\label{eq:mGF2:1}
			\mathcal{F}_{m,\otimes}^\alpha\mathbf{x}:=\dfrac{1}{2}\left(\begin{matrix}(\mathbf{P}_{m,\otimes}+\mathbf{Q}_{m,\otimes})^{T}\mathbf{x}\\
				(\mathbf{P}_{m,\otimes}-\mathbf{Q}_{m,\otimes})^{T}\mathbf{x} \end{matrix}\right).		
		\end{equation}
		Moreover, the inverse MGFRFT $\mathcal{F}_{m,\otimes}^{-\alpha}: \mathbb{R}^{2N}\mapsto \mathbb{R}^{N}$ is given by
		\begin{equation}\label{eq:mGF2:2}
			\mathcal{F}_{m,\otimes}^{-\alpha}
			\left(\begin{matrix}\mathbf{y}_{1}\\
				\mathbf{y}_{2}\end{matrix}\right)
			:=\frac{1}{2}[\mathbf{P}_{m,\otimes}(\mathbf{y}_{1}+\mathbf{y}_{2})+\mathbf{Q}_{m,\otimes}(\mathbf{y}_{1}-\mathbf{y}_{2})],
		\end{equation}
		where $\mathbf{y}_{1},\mathbf{y}_{2}\in\mathbb{R}^{N}$.
	\end{Definition}
	
%
	In addition, we show that the most of the energy of signals with strong spatial-temporal correlation on directed Cartesian product graph $\mathcal{G}$ with $m$ directed graphs is concentrated in the low frequencies of the MGFRFT $\mathcal{F}_{m,\otimes}^\alpha$.
	\begin{Theorem}\label{Thm:mAn2}
		Assume that $\mathcal{G}=\mathcal{G}_{1}\boxtimes\mathcal{G}_{2}\boxtimes\cdots\boxtimes\mathcal{G}_{m}$ is a Cartesian product graph of $m$ directed graphs $\mathcal{G}_{l},\;l=1,2,\cdots,m$, the fractional Laplacian matrices $\mathbf{L}_{l}^{\alpha}$, and $\mathbf{p}_{l,i},\, \mathbf{q}_{l,i}$, $r_{l,i},\, 0\le i\le N_{l}-1,\,l=1,2,\cdots, m$ are defined the same as $(\ref{eq:mLa})$, $\alpha$ is the fractional order with $0<\alpha\le 1$, $\tau_{m,k},\, 0\leq k\leq N-1$, are sorted in ascending order of $r_{1,i_1}+r_{2,i_2}+\cdots+r_{m,i_m},\,0\leq i_l\leq N_{l}-1,\;l=1,2,\cdots, m,$ and $N=N_{1}N_{2}\cdots N_{m}$. Let $\Gamma\in\{1,2,\cdots,N\}$ be the frequency bandwidth of the GFRFT $\mathcal{F}_{m,\otimes}^\alpha$ in $(\ref{eq:mGF2:1})$, and the low frequency component of a signal $\mathbf{x}$ on $\mathcal{G}$ be
		\begin{align*}\label{chuwz:53}
			\mathbf{x}_{\Gamma,m,\otimes}^\alpha&=\frac{1}{2}\sum_{(i_1,i_2,\cdots,i_m)\in \mathcal{S}_{m,\Gamma}}\Big[(\mathbf{p}_{1,i_1}\otimes\mathbf{p}_{2,i_2}\otimes\cdots\otimes\mathbf{p}_{m,i_m})\\
			&\quad\quad\times(\mathbf{p}_{1,i_1}\otimes\mathbf{p}_{2,i_2}\otimes\cdots\otimes\mathbf{p}_{m,i_m})^{T}\mathbf{x}
			+(\mathbf{q}_{1,i_1}\otimes\mathbf{q}_{2,i_2}\otimes\cdots\otimes\mathbf{q}_{m,i_m})\\
			&\quad\quad\quad\times(\mathbf{q}_{1,i_1}\otimes\mathbf{q}_{2,i_2}\otimes\cdots\otimes\mathbf{q}_{m,i_m})^{T}\mathbf{x}\Big],
		\end{align*}
		where $\mathcal{S}_{m,\Gamma}=\{(i_1,i_2,\cdots,i_m)|\tau_{m,k}=r_{1,i_1}+r_{2,i_2}+\cdots+r_{m,i_m},0\le k\le\Gamma-1\}$. Then, we have
		\begin{align*}
		&\|\mathbf{x}-\mathbf{x}_{\Gamma,m,\otimes}^{\alpha}\|_{2}\\
		\leq&\frac{1}{2\tau_{m,\Gamma-1}}\Big[\|(\mathbf{L}_1^{\alpha}\otimes \mathbf{I}_{N_2N_{3}\cdots N_m})\mathbf{x}\|_{2} +\|(\mathbf{I}_{N_1}\otimes\mathbf{L}_2^{\alpha}\otimes \mathbf{I}_{N_{3}N_{4}\cdots N_m})\mathbf{x}\|_{2}\\
		&\quad+\cdots+\|(\mathbf{I}_{N_{1}N_{2}\cdots N_{m-1}}\otimes\mathbf{L}_m^{\alpha})\mathbf{x}\|_{2}+\|((\mathbf{L}_{1}^{\alpha})^{T}\otimes\mathbf{I}_{N_2N_{3}\cdots N_m})\mathbf{x}\|_{2}\\
		&\quad+\|(\mathbf{I}_{N_1}\otimes(\mathbf{L}_{2}^{\alpha})^{T}\otimes \mathbf{I}_{N_{3}N_{4}\cdots N_m})\mathbf{x}\|_{2}+\cdots+\|(\mathbf{I}_{N_{1}N_{2}\cdots N_{m-1}}\otimes(\mathbf{L}_{m}^{\alpha})^{T})\mathbf{x}\|_{2}\Big],
	\end{align*}
		where $\tau_{m,\Gamma-1}$ is the cutoff frequency.
	\end{Theorem}
    The proof of Theorems \ref{Thm-mGF} and \ref{Thm:mAn2} are similar to those of Theorems \ref{Thm-GFRFT1} and \ref{Thm:An2}, respectively. Therefore, for the sake of brevity, we omit the proof.
	
	\section{Numerical Experiments}\label{sec6}
	
	In this section, compared with the DGFRFT $\mathcal{F}_{q}^{\alpha}$ (\ref{Def:spqd}) proposed in \cite{YL23}, better denoising performances of our two GFRFTs $\mathcal{F}_{\boxtimes}^{\alpha}$ and $\mathcal{F}_{\otimes}^{\alpha}$ are shown on the hourly temperature data set published by the French National Meteorological Service \cite{PV17}, which is collected from 32 weather stations in the Brest region of France on January 2014. The original temperature data is denoted as matrices $\mathbf{X}_{d}=[\mathbf{x}_{d}(t_{0}),\cdots,\mathbf{x}_{d}(t_{23})],\;1\leq d \leq31$, where
	$\mathbf{x}_{d}(t_{i}),\;0\leq i\leq23$, are column vectors, representing the temperatures of 32 weather stations at the time $t_{i}$ on day $d$ of January 2014.
	These data are available at \url{https://donneespubliques.meteofrance.fr/donnees_libres/Hackathon/RADOMEH.tar.gz}.
	In this experiment, we consider the denoising performances of three GFRFTs by bandlimiting the first $\Omega$ frequencies of the temperature data set with additive noise ${\boldsymbol{\epsilon}}_{d}$, i.e.,
	\begin{equation}\label{chuwz:36}
		\widehat{\mathbf{X}}_{d}=\mathbf{X}_{d}+{\boldsymbol{\epsilon}}_{d},1\leq d\leq 31,
	\end{equation}
	where the entries ${\boldsymbol{\epsilon}}_{d}$ are i.i.d., and obey the uniform distribution on the interval $[-\varepsilon, \varepsilon]$ with $\varepsilon\in[0,8]$. Let $\alpha=0.7$ and $q=1/2$ for DGFRFT $\mathcal{F}_q^{\alpha}$ throughout this section. In Figure \ref{fig:1}, we plot the original weather data set recorded in the region of Brest in France on January 2014, and the noisy data set is collected at noon on January 1st 2014, with noises following uniform distribution on $[-4,4]$.
	All numerical simulations are performed on a Thinkbook with Intel Core i7 -11800H  and  16GB  RAM, by MATLAB R2022a.
	
	\begin{figure}[htbp]
		\centering
		\subfigure[Original Signal]{
			\begin{minipage}[t]{0.48\linewidth}
				\centering
				\includegraphics[width=1.1\linewidth]{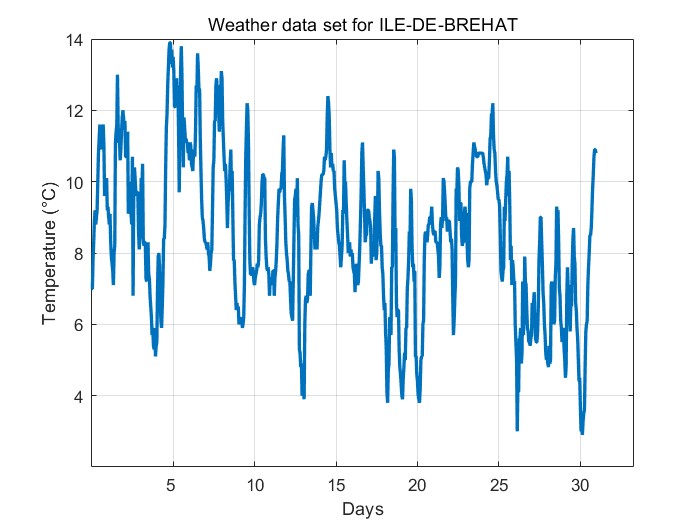}
			\end{minipage}%
		}%
		\subfigure[Noisy signal]{
			\begin{minipage}[t]{0.49\linewidth}
				\centering
				\includegraphics[width=1.1\linewidth]{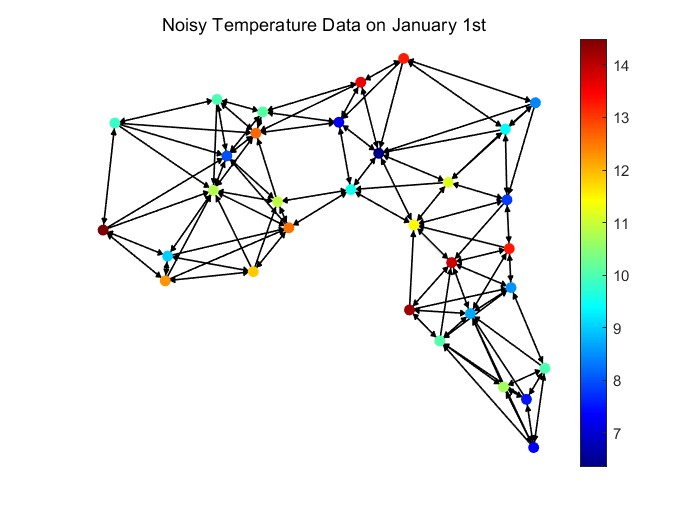}
			\end{minipage}%
		}%
		\centering
		\caption{Original and noisy temperature signals collected at 32 weather stations in the region of Brest in France, on January 2014.}\label{fig:1}
	\end{figure}
	
	We consider the matrices $\mathbf{X}_{d},\;1\leq d \leq31$ as signals defined on a Cartesian product graphs $\mathcal{T}\boxtimes\mathcal{S}$, where
	$\mathcal{T}$ represents an unweighted directed line graph with 24 nodes, $\mathcal{S}$ stands for a weighted directed graph with 32 locations of weather stations as nodes, and the edges are denoted as the 5 nearest neighboring stations based on the physical distances by 3 different weights $w_p, p=1,2,3$ that are constructed the same as \cite{CCL23}.
    Three edge weight schemes were employed for graph $\mathcal{S}$: (i) $w_1$ as i.i.d. random draws from Uniform $[0.8, 1.2]$; (ii) $w_2$ as absolute Pearson correlations between the $31$ feature matrices $\{\mathbf{X}_d\}_{d=1}^{31}$ for connected nodes, perturbed by additive uniform noise in $[-0.2, 0.2]$; and (iii) $w_3$ as absolute average temperature differences between adjacent nodes, with identical noise corruption.
Assume that $\mathbf{x}^{(i)}=(\mathbf{x}^{(i)}_{d}(t))_{1\leq d\leq31,\;0\leq t\leq23}$ is a vector constituting of weather data $\mathbf{x}^{(i)}_{d}(t)$ on $i$-th vertices at the $t$-th hour of $d$-th day. Then, three different types of weights on an edge from $j$ to $i$ are defined by
	\begin{equation}\label{chuwz:39}
		w_{1}(i,j)=1+u(i,j),
	\end{equation}
	\begin{equation}\label{chuwz:40}
		w_{2}(i,j)=\max \left(\frac{\left|\operatorname{Cov}\left(\mathbf{x}^{(i)}, \mathbf{x}^{(j)}\right)\right|}{\operatorname{Var}\left(\mathbf{x}^{(i)}\right)
			\operatorname{Var}\left(\mathbf{x}^{(j)}\right)}+u(i,j), 0\right),
	\end{equation}
	and
	\begin{equation}\label{chuwz:41}
		w_{3}(i,j)=\max\left(\left|\mathbb{E}(\mathbf{x}^{(i)})-\mathbb{E}(\mathbf{x}^{(j)})\right|+u(i, j), 0\right),
	\end{equation}
	where $u(i,j)$ are i.i.d. with uniform distribution on $[-0.2,0.2]$,\; $\mathbb{E}(\cdot)$,\; Var($\cdot$), and Cov$(\cdot, \cdot)$ represent mean, standard deviation, and covariance, respectively.
	
	At the beginning, on the directed Cartesian product graph $\mathcal{G}=\mathcal{T}\boxtimes \mathcal{S}$, we consider time costs on finding the left (or right) frequency components $\mathbf{p}_{k}$ (or $\mathbf{q}_{k}$), $0\leq k\leq767$ of $\mathcal{F}_{\boxtimes}^{\alpha}$, and $\mathbf{p}_{1,i}\otimes \mathbf{p}_{2,j}$ (or $\mathbf{q}_{1,i}\otimes \mathbf{q}_{2,j}$), $0\leq i\leq23,\,0\leq j\leq31$ of $\mathcal{F}_{\otimes}^{\alpha}$, and the frequency components $\mathbf{p}_{1,q,i}\otimes\mathbf{p}_{2,q,j},\,0\leq i\leq23,\,0\leq j\leq31$ of DGFRFT $\mathcal{F}_{q}^{\alpha}$ (\ref{Def:spqd}) with three types of weights. Compared to $\mathcal{F}_{\boxtimes}^{\alpha}$ and $\mathcal{F}_{q}^{\alpha}$,  $\mathcal{F}_{\otimes}^{\alpha}$ has lower computational complexity, which are illustrated in Table \ref{freq-cost}, the times are recorded in seconds.

	\begin{table}[ht!]
		\begin{center}
			\caption{Time cost for finding frequency components of $\mathcal{F}_{\boxtimes}^{\alpha}$,\;$\mathcal{F}_{\otimes}^{\alpha}$, and $\mathcal{F}_{q}^{\alpha}$ with three types of weights.}\label{freq-cost}
			\centering
			\begin{tabular}{cccc}
				\hline
				Weights & $\mathcal{F}_{\boxtimes}^{\alpha}$ & $\mathcal{F}_{\otimes}^{\alpha}$  & $\mathcal{F}_{q}^{\alpha}$ \\
				\hline
				$w_1$ &0.0643  & 0.0014 & 0.0031\\
				\hline
				$w_2$ &0.1295  &0.0018  & 0.0020\\
				\hline
				$w_3$ &0.1909  &0.0018  & 0.0019  \\
				\hline
			\end{tabular}
		\end{center}
		\vspace {-1.0em}
	\end{table}
	Next, we draw three GFRFTs $\mathcal{F}_{\boxtimes}^{\alpha}\mathbf{x}_{1}$, $\mathcal{F}_{\otimes}^{\alpha}\mathbf{x}_{1}$, and $\mathcal{F}_{q}^{\alpha}\mathbf{x}_{1}$ of graph signal $\mathbf{x}_{1}$ with weight $w_1$ in Figure \ref{fig:3}, where $\mathbf{x}_{1}$ is a vectorization of matrix $\mathbf{X}_{1}$.
	It is shown that approximately 88.66\%, 99.62\%, and 81.26\% of the energy of the temperature data $\mathbf{X}_{1}$ are concentrated in the first 40 of all 768 frequencies of $\mathcal{F}_{\boxtimes}^{\alpha}$, $\mathcal{F}_{\otimes}^{\alpha}$, and $\mathcal{F}_{q}^{\alpha}$, respectively. Similarly, by using the weight $w_{2}$ (or $w_3$), our numerical experiments indicate that the first 40 frequencies of $\mathcal{F}_{\boxtimes}^{\alpha}\mathbf{x}_{1}$, $ \mathcal{F}_{\otimes}^{\alpha}\mathbf{x}_{1}$, and $\mathcal{F}_{q}^{\alpha}\mathbf{x}_{1}$ containing about 90.05\%, 99.63\%, and 81.27\% (or 93.99\%, 99.62\%, and
	81.29\%) energy of $\mathbf{x}_{1}$, respectively. Therefore, our proposed two GFRFTs $\mathcal{F}_{\boxtimes}^{\alpha}$, $\mathcal{F}_{\otimes}^{\alpha}$ are more effective in representing temperature signal $\mathbf{x}_1$ than DGFRFT $\mathcal{F}_{q}^{\alpha}$.
	
	\begin{figure}
		\centering
		\subfigure[The first components in $\mathcal{F}_{\boxtimes}^{\alpha}\mathbf{x}_{1}$]{\begin{minipage}{0.49\linewidth}
				\centering
				\includegraphics[width=1\linewidth]{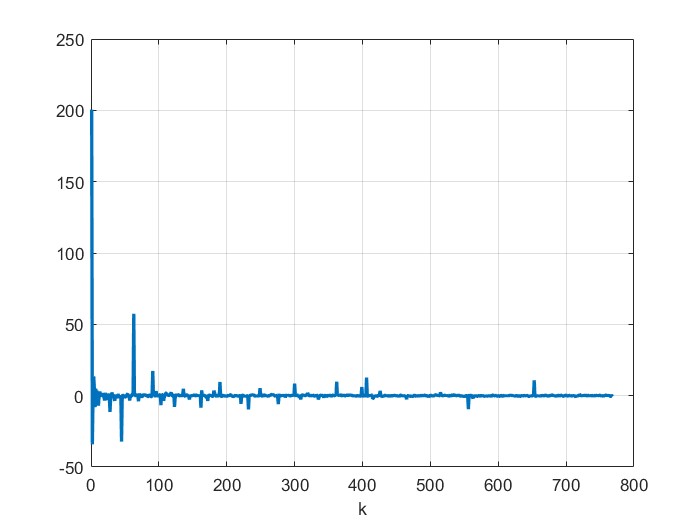}
		\end{minipage}}
		\centering
		\subfigure[The second components in $\mathcal{F}_{\boxtimes}^{\alpha}\mathbf{x}_{1}$]{\begin{minipage}{0.49\linewidth}
				\centering
				\includegraphics[width=1\linewidth]{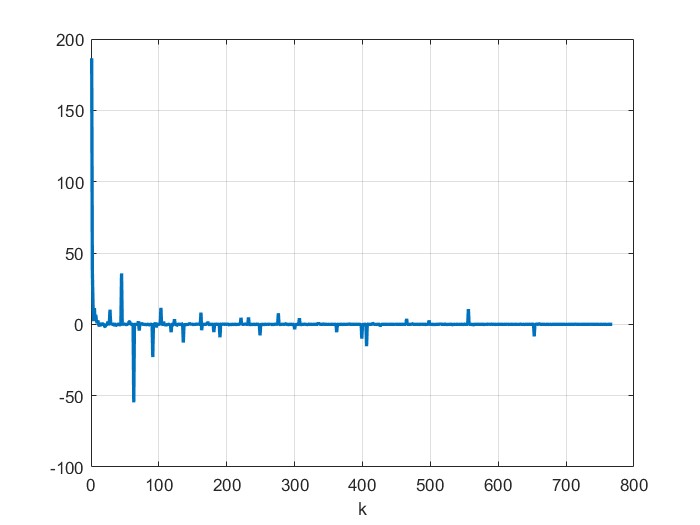}
		\end{minipage}}
		\centering
		\subfigure[The first components in $\mathcal{F}_{\otimes}^{\alpha}\mathbf{x}_{1}$]{\begin{minipage}{0.49\linewidth}
				\centering
				\includegraphics[width=1\linewidth]{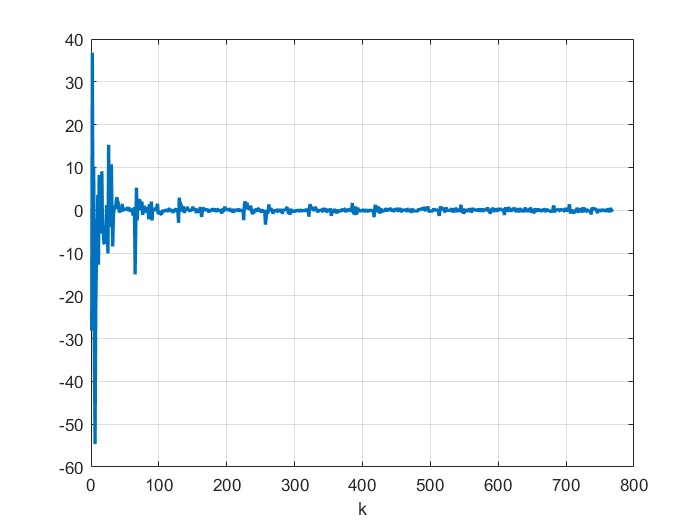}
		\end{minipage}}
		\centering
		\subfigure[The second components in $\mathcal{F}_{\otimes}^{\alpha}\mathbf{x}_{1}$]{\begin{minipage}{0.49\linewidth}
				\centering
				\includegraphics[width=1\linewidth]{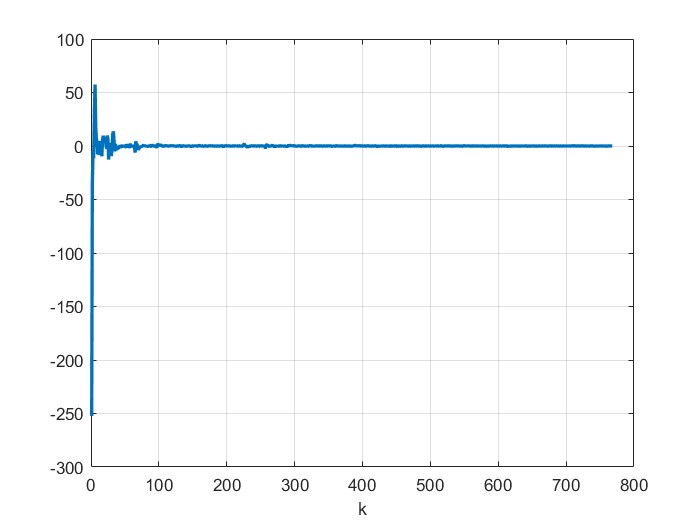}
		\end{minipage}}
		\centering
		\subfigure[The real part of $\mathcal{F}_{q}^{\alpha}\mathbf{x}_{1}$]{\begin{minipage}{0.49\linewidth}
				\centering
				\includegraphics[width=1\linewidth]{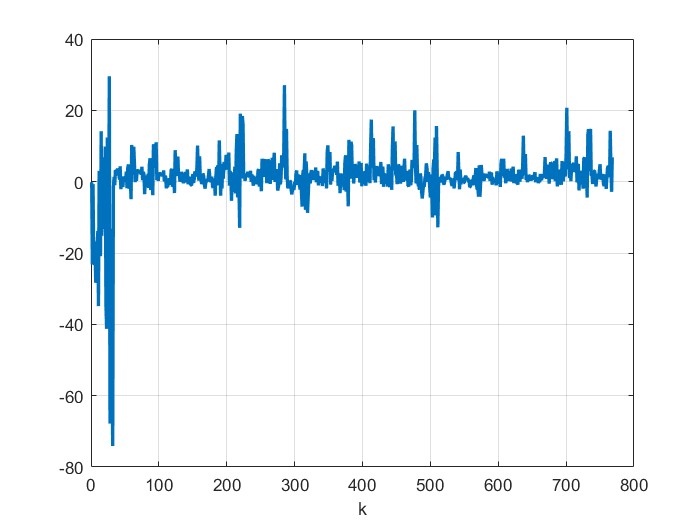}
		\end{minipage}}
		\centering
		\subfigure[The imaginary part of $\mathcal{F}_{q}^{\alpha}\mathbf{x}_{1}$]{\begin{minipage}{0.49\linewidth}
				\centering
				\includegraphics[width=1\linewidth]{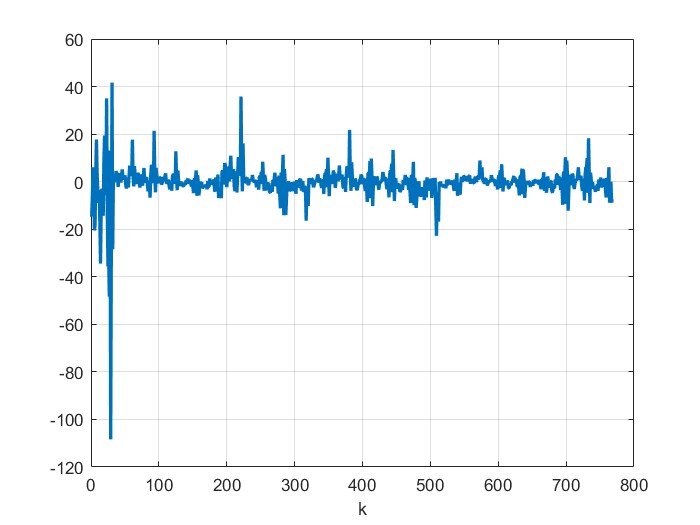}
		\end{minipage}}
		\caption{On the top are the first component $\frac{1}{2}(\mathbf{P}_{\boxtimes}+\mathbf{Q}_{\boxtimes})^{T}\mathbf{x}_1$ and the second component $\frac{1}{2}(\mathbf{P}_{\boxtimes}-\mathbf{Q}_{\boxtimes})^{T}\mathbf{x}_1$ of the GFRFT $\mathcal{F}_{\boxtimes}^{\alpha}\mathbf{x}_{1}$ (\ref{chuwz:12}). On the middle are the first component $\frac{1}{2}(\mathbf{P}_{\otimes}+\mathbf{Q}_{\otimes})^{T}\mathbf{x}_1$ and second component $\frac{1}{2}(\mathbf{P}_{\otimes}-\mathbf{Q}_{\otimes})^{T}\mathbf{x}_1$ of the GFRFT $\mathcal{F}_{\otimes}^{\alpha}\mathbf{x}_{1}$ (\ref{chuwz:21}). On the bottom are the real part and imaginary part of $(\mathbf{P}_{1,q}\otimes\mathbf{P}_{2,q})^*\mathbf{x}_1$ for the DGFRFT $\mathcal{F}_{q}^{\alpha}\mathbf{x}_{1}$ (\ref{Def:spqd}).}
		\label{fig:3}
	\end{figure}
	
	Moreover, we investigate the denoising performances of three GFRFTs $\mathcal{F}_{\boxtimes}^{\alpha}\mathbf{x}_{1}$, $\mathcal{F}_{\otimes}^{\alpha}\mathbf{x}_{1}$, and $\mathcal{F}_{q}^{\alpha}\mathbf{x}_{1}$ by using the bandlimiting processing to the noisy temperature data set $\widehat{\mathbf{X}}_1$ in (\ref{chuwz:36}), i.e., we only use the first $\Omega$-frequencies to recover the original signals $\mathbf{X}_1$. Let $1\le\Omega\le 768$, and let
	\begin{equation}\label{band:1}
		\mathbf{\widetilde{X}}_{1,\Omega,\boxtimes}^{\alpha}
		:={\rm{vec}^{-1}}\left(\frac{1}{2}\sum_{i=0}^{\Omega-1}(\mathbf{p}_{i}\mathbf{p}_{i}^{T}+\mathbf{q}_{i}\mathbf{q}_{i}^{T}){\rm{vec}}(\mathbf{\widehat{X}}_1)\right),
	\end{equation}
	\begin{eqnarray}\label{band:2}
		\mathbf{\widetilde{X}}_{1,\Omega,\otimes}^{\alpha}
		&:=&{\rm{vec}^{-1}}\Bigg(\frac{1}{2}\sum_{(i,j)\in \mathcal{S}_{\Omega}}\Big[(\mathbf{p}_{1,i}\otimes\mathbf{p}_{2,j})(\mathbf{p}_{1,i}\otimes\mathbf{p}_{2,j})^{T}{}\nonumber\\
		&&\times{\rm{vec}}(\mathbf{\widehat{X}}_1)
		+(\mathbf{q}_{1,i}\otimes\mathbf{q}_{2,j})(\mathbf{q}_{1,i}\otimes\mathbf{q}_{2,j})^{T}
		\times{\rm{vec}}(\mathbf{\widehat{X}}_1)\Big]\Bigg),
	\end{eqnarray}
	\begin{eqnarray}\label{band:3}
		\mathbf{\widetilde{X}}_{1,\Omega,q}^{\alpha}
		:={\rm{vec}^{-1}}\Bigg(\sum_{(i,j)\in \mathcal{U}_{\Omega}}(\mathbf{p}_{1,q,i}\otimes\mathbf{p}_{2,q,j})\times(\mathbf{p}_{1,q,i}\otimes\mathbf{p}_{2,q,j})^{*}{\rm{vec}}(\mathbf{\widehat{X}}_1)\Bigg),
	\end{eqnarray}
	where $\mathcal{S}_{\Omega}=\{(i,j)|\tau_k=r_{1,i}+r_{2,j},0\le k\le\Omega-1\}$, and $\mathcal{U}_{\Omega}=\{(i,j)|\mu_l=\varphi_{1,q,i}+\varphi_{2,q,j},0\le l\le\Omega-1\}$.

	\begin{figure}
		\centering
		\subfigure[Denoised signals $\widetilde{\mathbf{X}}_{1,\Omega,\boxtimes}^\alpha$ with weight $w_{1}$ ]{\begin{minipage}{0.48\linewidth}
				\centering
				\includegraphics[width=1\linewidth]{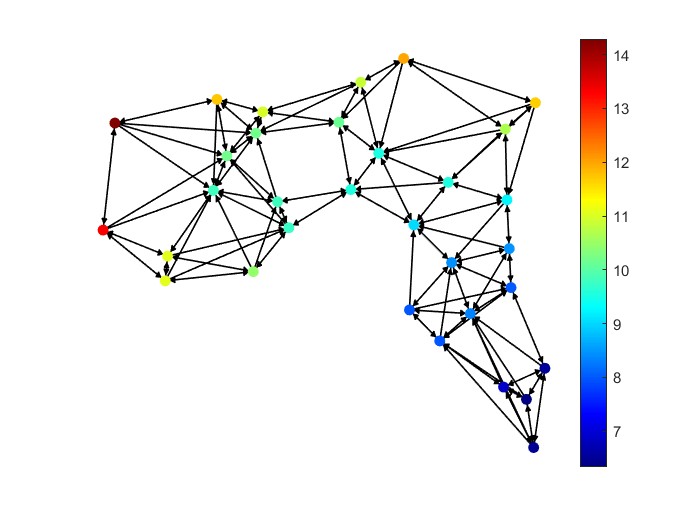}
		\end{minipage}}
		\centering
		\subfigure[Denoised signals $\widetilde{\mathbf{X}}_{1,\Omega,\boxtimes}^\alpha$ with weight $w_{2}$ ]{\begin{minipage}{0.48\linewidth}
				\centering
				\includegraphics[width=1\linewidth]{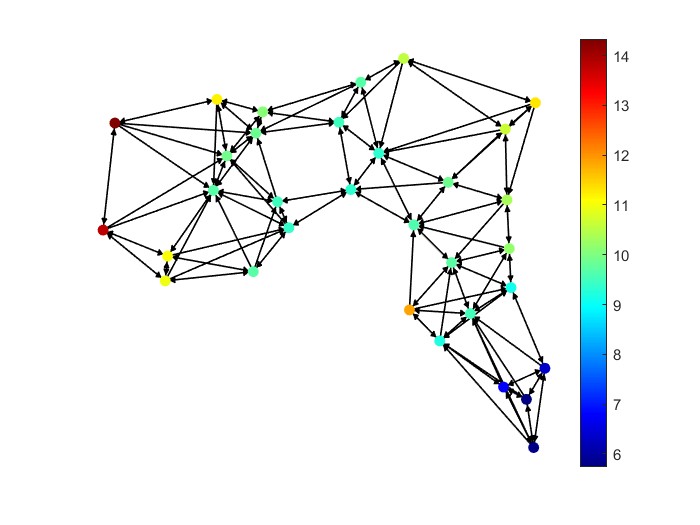}
		\end{minipage}}
		\centering
		\subfigure[Denoised signals $\widetilde{\mathbf{X}}_{1,\Omega,\otimes}^\alpha$ with weight $w_{1}$ ]{\begin{minipage}{0.48\linewidth}
				\centering
				\includegraphics[width=1\linewidth]{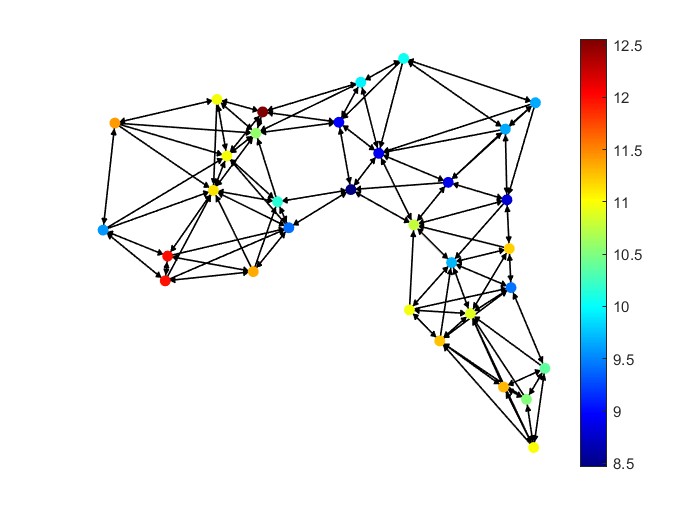}
		\end{minipage}}
		\centering
		\subfigure[Denoised signals $\widetilde{\mathbf{X}}_{1,\Omega,\otimes}^\alpha$ with weight $w_{2}$ ]{\begin{minipage}{0.48\linewidth}
				\centering
				\includegraphics[width=1\linewidth]{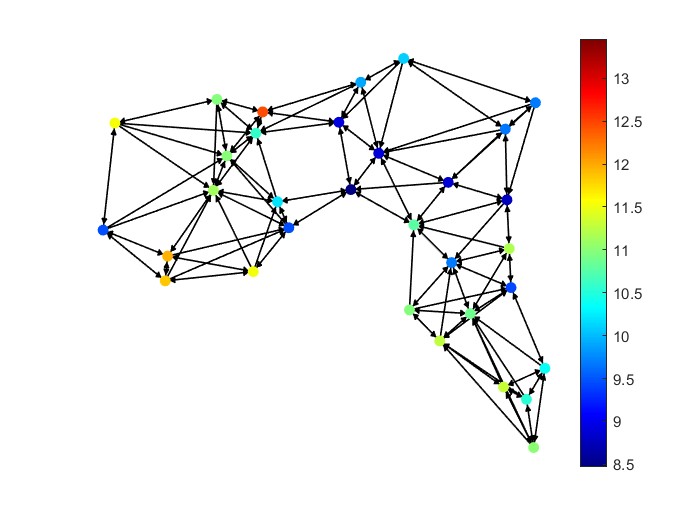}
		\end{minipage}}
		\centering
		\subfigure[Denoised signals $\widetilde{\mathbf{X}}_{1,\Omega,q}^\alpha$ with weight $w_{1}$ ]{\begin{minipage}{0.48\linewidth}
				\centering
				\includegraphics[width=1\linewidth]{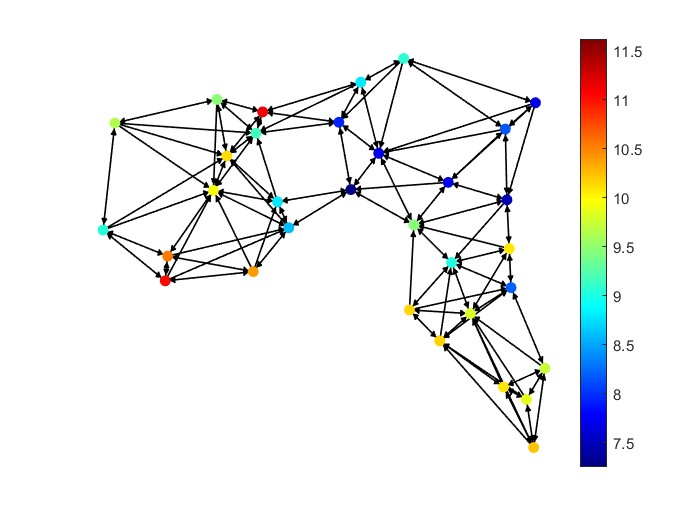}
		\end{minipage}}
		\centering
		\subfigure[Denoised signals $\widetilde{\mathbf{X}}_{1,\Omega,q}^\alpha$ with weight $w_{2}$]{\begin{minipage}{0.48\linewidth}
				\centering
				\includegraphics[width=1\linewidth]{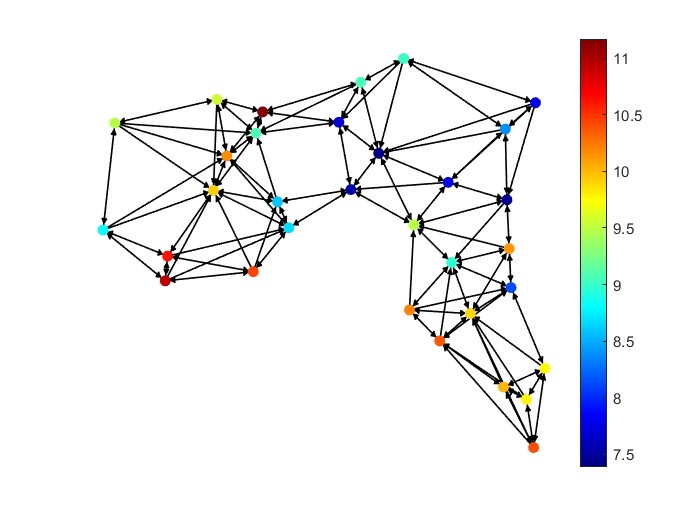}
		\end{minipage}}
		\caption{Denoised signals $\widetilde{\mathbf{X}}_{1,\Omega,\boxtimes}^\alpha$, $\widetilde{\mathbf{X}}_{1,\Omega,\otimes}^\alpha$ and $\widetilde{\mathbf{X}}_{1,\Omega,q}^\alpha$ with weights $w_{1}$ and $w_{2}$.}
		\label{fig:4}
	\end{figure}
	
	Let
	\begin{equation*}\label{chuwz:42}
		{\rm{ISNR}}(\varepsilon):=-20\log_{10} \frac{\|\mathbf{\widehat{X}}_1-\mathbf{X}_1\|_{F}}{\|\mathbf{X}_1\|_{F}},
	\end{equation*}
	\begin{equation*}\label{chuwz:43}
		{\rm{SNR}}(\varepsilon,\Omega):=-20 \log _{10} \frac{\|\mathbf{\widetilde{X}}_1-\mathbf{X}_1\|_{F}}{\|\mathbf{X}_1\|_{F}},
	\end{equation*}
	and
	\begin{equation*}\label{chuwz:44}
		{\rm{BAE}}(\varepsilon,\Omega):=\|\mathbf{\widetilde{X}}_1-\mathbf{X}_1\|_{\infty},
	\end{equation*}
	be the input signal-to-noise ratio (ISNR), the bandlimiting signal-to-noise ratio (SNR), and the bandlimiting approximation error (BAE), respectively. Here, $\mathbf{\widetilde{X}}_1$ is the bandlimited temperature data set of $\mathbf{X}_1$ in the form of (\ref{band:1}), or (\ref{band:2}), or (\ref{band:3}). Let the SNR and BAE derived from (\ref{band:1}), (\ref{band:2}), and (\ref{band:3})) be SNR$_{\boxtimes}$, SNR$_{\otimes}$ and SNR$_{q}$, and be BAE$_{\boxtimes}$, BAE$_{\otimes}$ and BAE$_{q}$, respectively. Let $\Omega=40$, in Figure $\ref{fig:4}$, we show three denoised signals $\widetilde{\mathbf{X}}_{1, \Omega,\boxtimes}^\alpha$, $\widetilde{\mathbf{X}}_{1, \Omega,\otimes}^\alpha$, and $\widetilde{\mathbf{X}}_{1,\Omega,q}^\alpha$ of noisy temperature data set $\widehat{\mathbf{X}}_{1}$ with respect to two weights $w_{1}$ and $w_{2}$, respectively. The corresponding SNRs and BAEs for bandlimiting approximation are listed in Table \ref{SNR}. 
Obviously, our two propose approach perform better on denoising than DGFRFT. Specially, $\mathcal{F}_{\otimes}^{\alpha}$ has best performance on recovery noisy signals among three GFRFTs.
	\begin{table}[ht!]
		\begin{center}
			\caption{The bandlimiting SNR and BAE for two weights $w_1$ and $w_2$.}\label{SNR}
			\vspace {1.0em}
            \centering
			\begin{tabular}{c|ccc|ccc}
				\hline
				Weights & SNR$_{\boxtimes}$ & SNR$_{\otimes}$  & SNR$_{q}$ & BAE$_{\boxtimes}$ & BAE$_{\otimes}$  & BAE$_{q}$ \\
				\hline
				$w_1$ &11.1217  & 20.9698 & 4.7456 &3.2736  & 0.8932 & 3.9921\\
				\hline
				$w_2$ &11.6522  &20.9252  & 4.7429 & 2.9165  &0.9001  & 3.9649\\
				\hline
			\end{tabular}
		\end{center}
		\vspace {-1.0em}
	\end{table}

Furthermore, we study the denoising performance of our proposed GFRFT methods under different noise levels $\varepsilon\in[0,8]$, different bandwidths $\Omega$, and different weights $w_i,i=1,2,3$, for a fixed fractional order $\alpha=0.7$. Specifically, ISNR, SNR$_{\boxtimes}$ , SNR$_{\otimes}$, SNR$_{q}$, BAE$_{\boxtimes}$, BAE$_{\otimes}$, and BAE$_{q}$ are each tested 100 times per day on average over a period of 31 days. From Tables \ref{ISNR}, \ref{3SNR}, and \ref{wiSNR}, we can see two points: (1) when denoising the noisy temperature dataset collected in the Brest region, our proposed GFRFTs $\mathcal{F}_{\boxtimes}^{\alpha}$ and $\mathcal{F}_{\otimes}^{\alpha}$ have better denoising performance than $\mathcal{F}_{q}^{\alpha}$ with respect to three different weights $w_i,i=1,2,3$, especially $\mathcal{F}_{\otimes}^{\alpha}$ has the best denoising effect. (2) When $\Omega\ge 40$, the SNRs of bandlimiting by GFRFT $\mathcal{F}_{\otimes}^{\alpha}$ changes slightly. The potential explanation is that the temperature data set in the Brest region of France exhibits a strong correlation across different hours and locations. Additionally, the energy of the original data set is predominantly concentrated in the low frequency components of the proposed GFRFT $\mathcal{F}_{\otimes}^{\alpha}$, as illustrated in the middle row of Figure \ref{fig:3}.
	\begin{table}[h]
		\begin{center}
			\caption{The average bandlimiting SNR and BAE for the weight $w_3$, with varying noise levels $\varepsilon\in[0, 8]$ and frequency bandwidth $\Omega=40$.}\label{ISNR}
\vspace {1.0em}
			\centering
			\begin{tabular}{c|c|ccc|ccc}
				\hline
				$\varepsilon$& ISNR &SNR$_{\boxtimes}$ &SNR$_{\otimes}$  &SNR$_{q}$ & BAE$_{\boxtimes}$ &BAE$_{\otimes}$  &BAE$_{q}$ \\
				\hline
				0&          $\infty$        &12.4190   &16.6353         &4.4379         & 2.9378       &1.5290      &3.6847\\
				\hline
				2&          17.1051       &12.3472         &16.3770        &4.4248    &2.9405      &1.5468      &3.6946\\
				\hline
				4&          11.0847       &12.1418        & 15.7370         &4.3853  &2.9555      &1.6239       &3.7340\\
				\hline
				6&          7.5655       &11.7976        & 14.9325       & 4.3212  &2.9678      &1.7752      &3.8025\\
				\hline
				8&          5.0621      & 11.3582        & 14.0388         &4.2318   &3.0072      &1.9981       &3.8980\\
				\hline
			\end{tabular}
		\end{center}
	\end{table}

	\begin{table}[h]
		\begin{center}
			\caption{The average bandlimiting SNR for the weight $w_3$, with noise level $\varepsilon=4$ and varying frequency bandwidth $\Omega$, where the average ISNR=10.0129.}\label{3SNR}
			\centering
			\begin{tabular}{cccc}
				\hline
				$\Omega$& SNR$_{\boxtimes}$ &SNR$_{\otimes}$  &SNR$_{q}$  \\
				\hline
				28&10.6586  &12.2602  &1.3892 \\
				\hline
				32&10.8395  &12.5766 &3.9530  \\
				\hline
				36&10.9002  &13.9150  &3.9566  \\
				\hline
				40&10.9597  &14.2123  &3.9604  \\
				\hline
				48&11.1163  &14.4032  &3.9681  \\
				\hline
				64&12.8102 &14.5860  &4.0845 \\
				\hline
			\end{tabular}
		\end{center}
	\end{table}
	\begin{table}[h]
		\begin{center}
			\caption{The average bandlimiting SNR for different weights $w_i,\;i=1,2,3$, with noise level $\varepsilon=4$ and varying frequency bandwidth $\Omega$, where the average ISNR=10.0129.}\label{wiSNR}
\vspace {1.0em}
			\centering
			\begin{tabular}{cccc}
				\hline
				Weights& SNR$_{\boxtimes}$ &SNR$_{\otimes}$  &SNR$_{q}$  \\
				\hline
				\multicolumn{4}{c}{$\Omega$=28}\\
				\hline
				$w_1$&9.5204  &12.1698  &1.8336  \\
				\hline
				$w_2$&9.5052  &12.2131 &1.7514 \\
				\hline
				$w_3$&10.6651 &12.2526 &1.3795  \\
				\hline
				\multicolumn{4}{c}{$\Omega$=32}\\
				\hline
				$w_1$&9.5169  &12.5700  &3.9529  \\
				\hline
				$w_2$&9.5141  &12.5700  &3.9529  \\
				\hline
				$w_3$&10.8351 &12.5700 &3.9529 \\
				\hline
				\multicolumn{4}{c}{$\Omega$=36}\\
				\hline
				$w_1$&9.5220  &13.7251  &3.9555  \\
				\hline
				$w_2$&9.5234  &13.7328  &3.9558  \\
				\hline
				$w_3$&10.9045  &13.9237 &3.9566 \\
				\hline
				\multicolumn{4}{c}{$\Omega$=40}\\
				\hline
				$w_1$&9.5168 &14.3760  &3.9609  \\
				\hline
				$w_2$&	9.8224 &	14.3675  &3.9599  \\
				\hline
				$w_3$&	10.9590 &	14.2206 &	3.9612 \\
				\hline
				\multicolumn{4}{c}{$\Omega$=48}\\
				\hline
				$w_1$&9.8017  &14.4623  &3.9698  \\
				\hline
				$w_2$&10.2516 &14.4775 &3.9689  \\
				\hline
				$w_3$&	11.1200  &	14.4024 &3.9680 \\
				\hline
				\multicolumn{4}{c}{$\Omega$=64}\\
				\hline
				$w_1$&10.5969  &14.5876  &4.0843  \\
				\hline
				$w_2$&	12.3486 &14.5876 &4.0843  \\
				\hline
				$w_3$&12.8325  &14.5876 &4.0843 \\
				\hline
			\end{tabular}
		\end{center}
	\end{table}

Finally, we explain the importance of fractional order $\alpha$. Table \ref{wialphaSNR} presents the bandlimiting SNR and BAE of our proposed GFRFTs $\mathcal{F}_{\boxtimes}^{\alpha}$ and $\mathcal{F}_{\otimes}^{\alpha}$ for different fractional orders $\alpha$ and weights  $w_i,i=1,2,3$. When $\alpha=1$, our proposed GFRFTs $\mathcal{F}_{\boxtimes}^{\alpha}$ and $\mathcal{F}_{\otimes}^{\alpha}$ reduce to GFTs $\mathcal{F}_{\boxtimes}$ (\ref{chuwz:4}) and $\mathcal{F}_{\otimes}$ (\ref{chuwz:7}) introduced by Cheng et al in \cite{CCL23}. It can be seen from Table \ref{wialphaSNR} that the SNRs for GFRFTs $\mathcal{F}_{\boxtimes}^{\alpha}$ and $\mathcal{F}_{\otimes}^{\alpha}$ with respect to different fractional orders and different weights $w_i,i=1,2,3$ are always higher than the results when $\alpha=1$, while the corresponding BAE is lower than the results when $\alpha=1$. This indicates that the denoising performance of our proposed GFRFTs $\mathcal{F}_{\boxtimes}^{\alpha}$ and $\mathcal{F}_{\otimes}^{\alpha}$ are better than those of GFTs $\mathcal{F}_{\boxtimes}$ and $\mathcal{F}_{\otimes}$ in \cite{CCL23}, and different selections of the factional order $\alpha$ give us greater flexibility in processing real-world data set.

	\begin{table}[h]
		\begin{center}
			\caption{The average bandlimiting SNR and BAE for different weights $w_i,\;i=1,2,3$, with noise level $\varepsilon=4$, frequency bandwidth $\Omega=40$, and varying fractional order $\alpha$, where the average ISNR=10.0129.}\label{wialphaSNR}
			\centering
			\begin{tabular}{c|cc|cc}
				\hline
				Weights& SNR$_{\boxtimes}$ &SNR$_{\otimes}$ & BAE$_{\boxtimes}$ &BAE$_{\otimes}$  \\
				\hline
				\multicolumn{5}{c}{$\alpha$=0.2}\\
				\hline
				$w_1$&12.5419  &14.3645 &2.0892  &1.4182  \\
				\hline
				$w_2$&12.2112  &14.3562&2.1067  &1.4155\\
				\hline
				$w_3$&12.2152  &14.2074 &2.0630  &1.4662 \\
				\hline
				\multicolumn{5}{c}{$\alpha$=0.5}\\
				\hline
				$w_1$&9.5400 &14.3501 &3.0054 &1.4088 \\
				\hline
				$w_2$&9.8705  &14.3444 &3.1503  &1.4082   \\
				\hline
				$w_3$&10.9621 &14.2415 &2.6465 &1.4409\\
				\hline
				\multicolumn{5}{c}{$\alpha$=0.8}\\
				\hline
				$w_1$&9.5274 &14.3560 &2.9996&1.4127\\
				\hline
				$w_2$&9.7327  &14.3507  &3.0905  &	1.4118  \\
				\hline
				$w_3$&10.9806 &14.2430 &2.6653 &1.4456 \\
				\hline
				\multicolumn{5}{c}{$\alpha$=1}\\
				\hline
				$w_1$&9.5203  &14.2656 &3.0129  &	1.4207\\
				\hline
				$w_2$&9.5846 &14.3217 &	3.1577 &1.4193\\
				\hline
				$w_3$&10.9614  &14.2124 &2.6692  &1.4643 \\
				\hline
			\end{tabular}
		\end{center}
	\end{table}

  In summary, our proposed GFRFTs $\mathcal{F}_{\boxtimes}^{\alpha}$ and $\mathcal{F}_{\otimes}^{\alpha}$ are capable of effectively decomposing graph signals defined on directed Cartesian product graphs into distinct frequency components, and can effectively process spatiotemporal signals with strong correlation. The GFRFT proposed in this paper is specifically designed for directed spatiotemporal signals with strong linear spatiotemporal correlation, which has limitations in processing complex nonlinear data. In order to be applicable to nonlinear signals, GFRFT can be combined with nonlinear extensions, such as kernelized graph fractional Laplacian operator \cite{GGK06}: replacing the linear fractional Laplacian operator with the kernel-induced fractional Laplacian operator to model the nonlinear similarity between nodes, while preserving the linear transformation structure of GFRFT in the kernel space; or a hybrid linear-nonlinear framework: using the proposed GFRFT for linear feature extraction, and then using nonlinear models (such as neural networks, kernel methods) to capture the residual nonlinear dependencies. These research contents are beyond the scope of this paper, but provide promising directions for future research.

	\section{Conclusion}\label{sec7}
	In this paper, based on SVD, we first propose two GFRFTs $\mathcal{F}_{\boxtimes}^{\alpha}$, $\mathcal{F}_{\otimes}^{\alpha}$ on Cartesian product graph of two directed graphs $\mathcal{G}_1$ and $\mathcal{G}_2$, and we prove that graph signals with strong spatial-temporal correlation can be stably recovered by our GFRFTs. In addtion, we extend our theoretical results to Cartesian product graph of $m$ directed graphs. Finally, experimental results show that our proposed GFRFTs have fairly good denoising performance, compared with DGFRFT $\mathcal{F}_{q}^{\alpha}$, DFTs $\mathcal{F}_{\boxtimes}$ and $\mathcal{F}_{\otimes}$. Especially, $\mathcal{F}_{\otimes}^{\alpha}$ takes lowest time on computing the frequencies of original signals, but has best reconstruction performance than other GFRFTs $\mathcal{F}_{\boxtimes}^{\alpha}$ and $\mathcal{F}_{q}^{\alpha}$.

\section*{Declaration of competing interest}
 The authors declare that they have no conflict of interest.

 \section*{Funding}
   This work is supported partially by the National Natural Science Foundation of China (Grant no. 12261059).
	
 \section*{Data availability}
Data will be made available on request.


\end{document}